\documentclass[a4paper]{amsart}


\usepackage{amssymb,latexsym}
\usepackage{amsmath,amsthm}
\usepackage{graphicx}
\usepackage{caption}
\usepackage{subcaption}
\usepackage[usenames,dvipsnames]{color}
\usepackage{float}
\usepackage{tikz,tikz-3dplot}
\usepackage{pgfplots}
\usepackage[backref=page,colorlinks=true,linkcolor=Orchid,citecolor=Orchid,urlcolor=Orchid]{hyperref}
\usepackage{marginnote}
\usepackage{enumerate}
\usepackage[makeroom]{cancel}

\newtheorem{theorem}{Theorem}

\newtheorem{proposition}{Proposition}
\newtheorem{lemma}{Lemma}
\newtheorem{AppLem}{Lemma}[section]

\newtheorem{AppDef}{Definition}[section]

\newtheorem{conjecture}{Conjecture}
\newtheorem{assumption}{Assumption}
\theoremstyle{remark}
\newtheorem{remark}{Remark}

\newcommand{\R}{\mathbb{R}}

\begin{document}


\title[LRS Bianchi III Einstein-Klein-Gordon future]{On the oscillations and future asymptotics of locally rotationally symmetric Bianchi type~III cosmologies with a massive scalar field}

\author{\textsc{David FAJMAN}}
\author{\textsc{Gernot HEI}\ss\textsc{EL}}
\author{\textsc{Maciej MALIBORSKI}}
\thanks{The authors acknowledge support of the Austrian Science Fund (FWF) through the Project \emph{Geometric transport equations and the non-vacuum Einstein flow} (P 29900-N27).}
\thanks{The second author also thanks the Hausdorff Research Institute for Mathematics for funding and hospitality in the Junior Trimester Program `Kinetic Theory' during the final research and write up phase of this project.}
\thanks{Finally, we are grateful to the referees for their detailed and helpful comments and suggestions.}

\address{\href{http://gravity.univie.ac.at}{Gravitational Physics Group\\Faculty of Physics\\University of Vienna\\Austria}}	


\maketitle

\begin{abstract}
We analyse spatially homogenous cosmological models of locally rotationally symmetric Bianchi type~III with a massive scalar field as matter model. Our main result concerns the future asymptotics of these spacetimes and gives the dominant time behaviour of the metric and the scalar field for all solutions for late times. This metric is forever expanding in all directions, however in one spatial direction only at a logarithmic rate, while at a power-law rate in the other two. Although the energy density goes to zero, it is matter dominated in the sense that the metric components differ qualitatively from the corresponding vacuum future asymptotics.

Our results rely on a conjecture for which we give strong analytical and numerical support. For this we apply methods from the theory of averaging in nonlinear dynamical systems. This allows us to control the oscillations entering the system through the scalar field by the Klein-Gordon equation in a perturbative approach.
\end{abstract}

\section{Introduction}\label{S: Intro}

Spatially homogenous cosmologies have been investigated in vacuum \cite{Ringstroem2013} and with various matter models such as perfect fluids \cite{Coley2003, WainwrightEllis1997}, collisionless kinetic gases (Vlasov matter) \cite{Andreasson2011, Rendall2004, Rendall2008} (recently \cite{BarzegarEtAl2020, FajmanHeissel2019, LeeEtAl2019}) magnetic fields or scalar fields \cite{Coley2003, Rendall2008} (recently \cite{FajmanWyatt2019, IonescuPausader2019, Wang2019}), in particular with a focus on their asymptotic dynamics and stability. A fruitful approach has been to formulate the respective systems in terms of expansion normalised (or Hubble normalised) dimensionless quantities. Typically the Raychaudhuri equation decouples, which is the evolution equation for the Hubble scalar. The latter describes the overall rate of spatial expansion. The asymptotics of the remaining reduced system is then typically given by equilibrium points and often can be determined by a dynamical systems analysis; cf \cite{Coley2003, Perko2001, WainwrightEllis1997}.

With the present work we contribute to spatially homogenous scalar field cosmology. In the literature a lot of attention has been given to scalar fields with an exponential potential. One of the reasons is that in this case the above approach is applicable since the Raychaudhuri equation decouples due to the symmetry that the exponential is its own derivative; cf \cite[\textsc{Sec}~IV, A]{Coley2003}. Scalar fields with harmonic potentials, such as massive scalar fields which satisfy the Klein-Gordon equation, have also been investigated \cite[\textsc{Sec}~III, E]{Coley2003}. Here however the Raychaudhuri equation fails to decouple, which is why the above approach is not applicable and the resulting systems prove to be difficult to analyse in a dynamical systems approach. On the other hand, other approaches, such as a local stability analysis, are also difficult to apply due to oscillations entering the system via the Klein-Gordon equation. Nevertheless there have been successful enquiries going this route, such as the recent works on various nonlinear stability problems for the Einstein-Klein-Gordon system concerning Minkowski spacetime \cite{IonescuPausader2019, LeFlochMa2016, Wang2016} and the Milne model \cite{FajmanWyatt2019, Wang2019}. While those results concern vacuum dominated future asymptotics, in the present work we consider a class of models where the future is \emph{matter dominated}. By this we mean that the dominant behaviour of the metric for large times differs qualitatively from that of the corresponding vacuum spacetime.

In this paper we take a different approach and apply techniques from the \emph{theory of averaging} in non-linear dynamical systems~\cite{SandersEtAl2007}. The core idea is to construct a time average of the original system, in the solution of which the oscillations are smoothed out. The oscillations are thus viewed as perturbations and it turns out that the Hubble scalar plays the role of a time dependent perturbation parameter which controls the magnitude of error between the full and the averaged solutions. Most importantly, the Raychaudhuri equation of the averaged system decouples, and the remaining reduced system can be analysed with the well known dynamical systems techniques mentioned above. Consequently we reduce the analysis of the original system to that of a corresponding averaged system.

In general the averaged solution represents only an approximation to the full solution, and error estimates in terms of the perturbation parameter are typically only valid for limited time. However in cases where either system is attracted to an equilibrium point the validity of such error estimates can be extended to all times; cf~\cite[\textsc{Chap}~5]{SandersEtAl2007}. It turns out that we are looking at such a case, and since we can show that our perturbation parameter, the Hubble scalar, goes to zero the full and the averaged solutions converge in the limit when time goes to infinity. The techniques presented here have been developed independently of~\cite{AlhoEtAl2019, AlhoEtAl2015, AlhoUggla2015} who took a closely related approach.

The spatially homogenous cosmologies we are considering are those of locally rotationally symmetric (LRS) Bianchi type~III, which in different terminology are also referred to as LRS Bianchi type~VI$_{-1}$; cf \cite{Rendall2004, RyanShepley1975, WainwrightEllis1997}. As matter model we consider a massive scalar field. Hence we are dealing with the Einstein-Klein-Gordon system within the LRS Bianchi III symmetry class. Our main results consider the future asymptotics of these systems. We express these in terms of the dominant behaviour of physical and geometric quantities for large times, such as the shear parameter and the energy density (\textsc{Theorem}~\ref{T: X asympt}) as well as the metric components (\textsc{Theorem}~\ref{T: g asympt}). Furthermore we also give the dominant late time behaviour of the Klein-Gordon field (\textsc{Theorem}~\ref{T: g asympt}).The premises of the main theorems contain the assumption that a conjecture holds (\textsc{Conjecture}~\ref{CJ: 1}). For the latter we give strong analytical and numerical support. Apart from supporting the conjecture, the numerics also convincingly demonstrate agreement with all our analytical results.

Finally we mention that one of our motivations to consider LRS Bianchi III Einstein-Klein-Gordon cosmologies was the result of~\cite{Rendall2002} which concerns Einstein-Vlasov cosmologies with massive particles within the same symmetry class. The finding was a forever expanding matter dominated future asymptotics with one scale factor increasing linearly with time, and the other one increasing slower than any power law. We were interested if the case of a massive scalar field would yield the same future asymptotics. Our results give an affirmative answer to this question.  \\

We begin with a brief introduction to periodic averaging in~\textsc{Section}~\ref{S: periodic averaging}. In \textsc{Section}~\ref{E: BIII EKG system} we introduce the LRS Bianchi III Einstein-Klein-Gordon system and formulate an averaging conjecture for it; cf \textsc{Conjecture}~\ref{CJ: 1}. The conjecture gives error estimates comparing the full and a corresponding averaged solution. We then provide analytical support for this conjecture in \textsc{Section}~\ref{S: analytical support}. Under the assumption that the conjecture holds, we then derive our main results in \textsc{Section}~\ref{S: asymptotics}: the future asymptotics of LRS Bianchi~III Einstein-Klein-Gordon cosmologies; cf theorems~\ref{T: X asympt} and~\ref{T: g asympt}. \textsc{Section}~\ref{S: numerical support} is then concerned with numerical support for our results and assumptions. Finally, we conclude with a discussion of our results and an outlook on open problems and further applicability of the averaging techniques developed here to the field of mathematical cosmology in \textsc{Section}~\ref{S: discussion}.

In particular for \textsc{Section}~\ref{S: analytical support} we assume some familiarity with standard methods of dynamical systems theory; cf~\cite{Perko2001}. However, because it has not been used to a wider extend in the literature of mathematical cosmology, we put a brief introduction to centre manifold analysis in \textsc{Appendix}~\ref{A: CM analysis}, following~\cite[\textsc{Chap~1}]{Carr1981}. We apply these tools in~\textsc{Subsection}~\ref{SS: CM analysis of D}. \textsc{Appendix}~\ref{A: two definitions} contains two technical definitions which we refer to in \textsc{Section}~\ref{S: periodic averaging}.

\section{Periodic averaging and the Van der Pol equation}\label{S: periodic averaging}

We restrict the following outline to a brief summary of those aspects of the theory of averaging which are relevant to the analysis of the present work, and follow Sanders, Verhulst and Murdock~\cite{SandersEtAl2007}, in particular its chapters~2 and~5.

In \textsc{Subsection}~\ref{SS: periodic averaging} we outline the basic idea of periodic averaging. We specify what we mean by a problem in standard form and for such give two theorems which estimate the error between the full and the corresponding averaged solution. The theory is then applied to the Van der Pol equation in \textsc{Subsection}~\ref{SS: Van der Pol}. This is not only an illustrative example, but is also of direct relevance to our problem since the spatially homogenous Klein-Gordon equation, though coupled to the Einstein equations, can be considered as a specific example of a Van der Pol equation.

\subsection{Periodic averaging}\label{SS: periodic averaging}

The theory of averaging studies initial value problems of the general form
\begin{align*}
\dot{\mathbf x} = \mathbf{f}(\mathbf x, t, \epsilon), \quad \mathbf x(0)=\mathbf a,
\end{align*}
with $\mathbf x, \mathbf f(\mathbf x, t, \epsilon)\in\R^n$, where $\epsilon$ plays the role of a, usually small, perturbation parameter. Typically one would then perform a Taylor expansion of $\mathbf f$ in $\epsilon$ around $\epsilon=0$. For the simplest form of averaging, \emph{periodic averaging}, the zeroth order term usually vanishes, and one is typically looking at problems of the standard form
\begin{align}\label{E: standard form}
\dot{\mathbf x} = \epsilon\,\mathbf f^1(\mathbf x, t) + \epsilon^2\,\mathbf f^{[2]}(\mathbf x,t,\epsilon),
\quad
\mathbf x(0) = \mathbf a,
\end{align}
with $\mathbf f^1$ and $\mathbf f^{[2]}$ $T$-periodic in $t$. The exponents correspond to the respective perturbative order, and the square bracket marks the remainder of the series; cf \cite[p~13, \sc Notation~1.5.2]{SandersEtAl2007}.

To first order, the theory is then concerned with the question to what degree solutions of~\eqref{E: standard form} can be approximated by the solutions of an associated \emph{averaged system}
\begin{align}\label{E: averaged system}
\dot{\mathbf z} &= \epsilon\,\overline{\mathbf f}^1(\mathbf z), \quad \mathbf z(0)=\mathbf a,
\end{align}
with
\begin{align}\label{E: f bar}
\overline{\mathbf f}^1(\mathbf z) &= \frac{1}{T}\int_0^T\mathbf f^1(\mathbf z, s)\,\mathrm ds.
\end{align}
The basic result is given by the following theorem:
\begin{lemma}[\text{\cite[p~31, \sc Thm~2.8.1]{SandersEtAl2007}}]\label{L: averaging}
Let $\mathbf f^1$ be Lipschitz continuous, let $\mathbf f^{[2]}$ be continuous, and let $\epsilon_0, D, L$ be as in \textsc{Definition}~\ref{D: D}. Then there exists a constant $c>0$ such that
\begin{align*}
||\mathbf x(t,\epsilon)-\mathbf z(t,\epsilon)|| &< c\epsilon
\end{align*}
for $0\leq\epsilon\leq\epsilon_0$ and $0\leq t\leq L/\epsilon$, and where $||\,.\, ||$ denotes the norm $||\mathbf u ||:=\sum_{i=1}^n|u_i|$ for $\mathbf u\in\R^n$.
\end{lemma}
In other words, the error one makes by approximating the full system~\eqref{E: standard form} by the averaged system~\eqref{E: averaged system} will be of order $\epsilon$ on timescales of order $\epsilon^{-1}$. When the solutions of the full or averaged system are attracted by an asymptotically stable critical point, the domain of approximation might be extendable to all times; cf~\cite[\sc Chap~5]{SandersEtAl2007}. For instance:
\begin{lemma}[\textbf{Eckhaus/Sanchez-Palencia \cite[p~101, \sc Thm~5.5.1]{SandersEtAl2007}}]\label{L: Eckhaus}
Consider the initial value problem
\begin{align*}
\dot{\mathbf x} &= \epsilon\,\mathbf f^1(\mathbf x,t), \quad \mathbf x(0)=\mathbf a,
\end{align*}
with $\mathbf a, \mathbf x\in D\subset\R^n$ and $\mathbf f^1$ $T$-periodic in $t$. Suppose $\mathbf f^1$ is a KBM-vector field (\textsc{Definition}~\ref{D: KBM}) producing the averaged equation
\begin{align*}
\dot{\mathbf z} &= \epsilon\,\overline{\mathbf f}^1(\mathbf z), \quad \mathbf z(0)=\mathbf a,
\end{align*}
where $\mathbf z=0$ is an asymptotically stable critical point in the linear approximation, $\overline{\mathbf f}^1$ is continuously differentiable with respect to $\mathbf z$ in $D$ and has a domain of attraction $D^o\subset D$. Then for any compact $K\subset D^o$ and for all $\mathbf a\in K$
\begin{align*}
\mathbf x(t)-\mathbf z(t) &= \mathcal O(\epsilon), \quad 0\leq t<\infty.
\end{align*}
\end{lemma}

\subsection{The Van der Pol equation}\label{SS: Van der Pol}

A demonstrative example of first order periodic averaging, which is relevant to the analysis of the present paper, is given by the class of Van der Pol equations
\begin{align}\label{E: Van der Pol}
\ddot\phi+\phi &= \epsilon\,g(\phi,\dot\phi)
\end{align}
with $g$ sufficiently smooth. \eqref{E: Van der Pol} describes a harmonic oscillator with generally non-linear feedback and damping. An amplitude-phase (variation of constants) transformation yields a system in standard form~\eqref{E: standard form}
\begin{align}\label{E: amplitude-phase}
\begin{matrix}
\phi = r\sin(t-\varphi) \\
\dot\phi = r\cos(t-\varphi)
\end{matrix} \quad \Longrightarrow \quad
\begin{bmatrix} \dot r \\ \dot\varphi \end{bmatrix} &= \epsilon
\begin{bmatrix}
\phantom{\frac{1}{r}}\cos(t-\varphi)\,g(\phi,\dot\phi) \\
\frac{1}{r}\sin(t-\varphi)\,g(\phi,\dot\phi)
\end{bmatrix}
\end{align}
where the arguments of $g$ are understood to be substituted by the transformation. Consequently, the averaged system~\eqref{E: averaged system} is given by\footnote{We follow here the notation of~\cite{SandersEtAl2007} and put bars above the variables to indicate that they belong to the averaged system.}
\begin{align}
\begin{bmatrix} \dot{\overline r} \\ \dot{\overline\varphi} \end{bmatrix} &= \epsilon\,\overline{\mathbf f}^1(\overline r) = \epsilon
\begin{bmatrix}
\overline f	^1_r(\overline r) \\
\overline f^1_\varphi(\overline r)
\end{bmatrix}
\end{align}
with
\begin{align}
\overline f^1_r(\overline r) &= \frac{1}{2\pi} \int_0^{2\pi} \cos(s-\overline\varphi)\,g\big(\overline r\sin(s-\overline\varphi), \overline r\cos(s-\overline\varphi)\big)\,\mathrm ds
\end{align}
and $\overline f^1_\varphi(r)$ defined analogously. Note that $\overline{\mathbf f}^1$ is independent of $\overline\varphi$. Specialising now to the specific Van der Pol equation with $g(\phi, \dot\phi) = (1-\phi^2)\dot\phi$ as an illustrative example we obtain the averaged system
\begin{align}
\begin{bmatrix} \dot{\overline r} \\ \dot{\overline\varphi} \end{bmatrix} &=
\begin{bmatrix}
\frac{1}{2}\epsilon\overline r (1-\frac{1}{4}\overline r^2) \\
0
\end{bmatrix}.
\end{align}
By \textsc{Lemma}~\ref{L: averaging} we know that the error between $[r, \varphi]^{\mathrm T}$ and $[\overline r, \overline\varphi]^{\mathrm T}$ will be of order $\epsilon$ on timescales of order $\epsilon^{-1}$. Since $\overline{\mathbf f}^1$ is independent of $\overline\varphi$ we restrict to the decoupled equation $\dot{\overline r}$,\footnote{In full rigorosity one has to justify this, which can be done by using a coordinate transformation $\phi=r\sin(\theta), \dot\phi=r\cos(\theta)$ instead, which yields a single averaged equation $\mathrm d\overline r/\mathrm d\theta=\epsilon f^1_1(\overline r)$, with $f^1_1(\overline r)$ of the same form as here; cf~\cite[p~103]{SandersEtAl2007}.} which has the two equilibrium points $\overline r=0$ and $\overline r=2$. The equilibrium point $\overline r=2$ is stable, and we can apply \textsc{Lemma}~\ref{L: Eckhaus} to extend the validity of the $\mathcal O(\epsilon)$ error estimate to all times into the future. (Similarly, by defining a negative time variable, we could apply the theorem to the past attraction to the unstable critical point at the origin.)

\section{The LRS Bianchi III Einstein-Klein-Gordon system in quasi-standard form}\label{E: BIII EKG system}

In this section we lay out the basic setup for our analysis. We introduce the LRS Bianchi~III Einstein-Klein-Gordon system in \textsc{Subsection}~\ref{SS: basic system}. In \textsc{Subsection}~\ref{SS: quasi-standard form} we then bring this system into a form closely resembling the standard form of periodic averaging problems discussed in the previous section. We call this the \emph{quasi-standard form} and formulate an averaging conjecture for it in \textsc{Subsection}~\ref{SS: averaging conjecture} which seeks to widen the applicability of the lemmas of the previous section to our case. At the end of the section we also derive the monotonicity and future asymptotics for one of our dynamical quantities, the Hubble scalar; cf \textsc{Lemma}~\ref{L: H to 0}. Finally we state the implications of the latter in conjunction with our averaging conjecture; cf \textsc{Proposition}~\ref{P: error limit}.

\subsection{The basic system}\label{SS: basic system}

An LRS Bianchi~III metric can be written in the form
\begin{align}\label{E: BIII metric}
\mathbf g &= -\mathrm dt^2 + a(t)^2\,\mathrm dr^2 + b(t)^2\,\mathbf g_{H^2}
\end{align}
where $\mathbf g_{H^2}$ denotes the 2-metric of negative constant curvature on hyperbolic 2-space; cf~\cite[\sc App~B.4]{CalogeroHeinzle2011}. We adopt the frame choice and variables of~\cite[\sc Sec~2]{RendallUggla2000} and~\cite[\sc Sec~3]{Rendall2002} and the resulting dynamical systems formulation of the Einstein equations for the metric~\eqref{E: BIII metric} and an energy-momentum tensor of the form $[{T^a}_b]=\mathrm{diag}(-\rho,p,p,p)$. For a Klein-Gordon field of mass $1$ and the metric~\eqref{E: BIII metric} we have
\begin{align}
\rho &= \frac{1}{2}\big(\dot\phi^2 + \phi^2\big) \quad\text{and}\quad
p = \frac{1}{2}\big(\dot\phi^2 - \phi^2\big)
\end{align}
where the field $\phi$ is subject to the Klein-Gordon equation $\Box_\mathbf{g}\phi=\phi$; cf~\cite[Sec~3.1]{Rendall2008}. Note that due to spatial homogeneity, the Einstein equations force $\phi$ to be independent of the spatial coordinates. Some of the equations are decoupled. For our analysis it suffices to restrict to the reduced coupled part of the LRS Bianchi~III Einstein-Klein-Gordon system which is given by
\begin{align}
\dot H &= H^2[-(1+q)] \label{E: Raychaudhuri} \\
\dot \Sigma_+ &= H[-(2-q)\Sigma_++1-\Sigma_+^2-\Omega] \label{E: Sprime} \\
\ddot \phi + \phi &= H[-3\dot\phi], \label{E: KG}
\end{align}
with the \emph{deceleration parameter}
\begin{align}
q &= 2\Sigma_+^2 + \frac{1}{6H^2}\big(2\dot\phi^2 - \phi^2\big).
\end{align}
$H := (\frac{\dot a}{a} + 2\frac{\dot b}{b})/3$ denotes the \emph{Hubble scalar}, ie a measure of the overall isotropic rate of spatial expansion. The corresponding evolution equation~\eqref{E: Raychaudhuri} is also referred to as the \emph{Raychaudhuri equation}. $H\Sigma_+ := (\frac{\dot a}{a} - \frac{\dot b}{b})/3$ denotes the only independent component of the \emph{shear tensor}, ie a measure of anisotropy in the rate of spatial expansion. Finally, $\Omega := \rho/(3H^2)$ defines a rescaled energy density which is non-negative by definition. Since we are interested in non-vacuum solutions we consider $\Omega$ to be positive.

The variables are subject to the \emph{Hamiltonian constraint}
\begin{align}\label{E: hc}
1-\Sigma_+^2-\Omega > 0.
\end{align}
Consequently, $\Sigma_+$ and $\Omega$ are bounded and take values in the range $(-1,1)$ and $(0,1)$ respectively. The Klein-Gordon field $\phi$ is unbounded, and so is $H$ a priori. However we will restrict our interest to the case of an (initially) expanding universe, ie to $H > 0$.

\subsection{The system in quasi-standard form}\label{SS: quasi-standard form}

The Klein-Gordon equation~\eqref{E: KG} has the form of a Van der Pol equation~\eqref{E: Van der Pol} with $g(\phi, \dot\phi) = -3\dot\phi$ and where $H(t)$ plays the role of $\epsilon$. The system~\eqref{E: Raychaudhuri}--\eqref{E: KG} therefore describes a harmonic oscillator with nonlinear damping, and where the time dependence of the latter is governed by the coupling of the Einstein equations with the Klein-Gordon equation via $H$.

Following \textsc{Subsection}~\ref{SS: Van der Pol} we apply an amplitude-phase transformation~\eqref{E: amplitude-phase} to formulate the Klein-Gordon equation~\eqref{E: KG} as two first order equations. Furthermore, we transform
\begin{align}\label{E: r to O}
r\mapsto\Omega = r^2/(6H^2) \quad\Longrightarrow\quad
q = 2\Sigma_+^2 + \Omega\big(3\cos(t-\varphi)^2-1\big).
\end{align}
The result is the system~\eqref{E: Raychaudhuri}, \eqref{E: Sprime} together with the first order Klein-Gordon equations
\begin{align}
\dot\Omega &= H\big[2\Omega \big(1+q-3\cos(t-\varphi)^2\big)\big] \label{E: Oprime} \\
\dot\varphi &= H[-3\sin(t-\varphi)\cos(t-\varphi)] \label{E: phiprime}
\end{align}
and where now $q$ is understood to be expressed by~\eqref{E: r to O}.

Setting $\mathbf x=[\Sigma_+,\Omega,\varphi]^{\mathrm T}$ this system has the form
\begin{align}\label{E: quasi-standard form}
\begin{bmatrix} \dot H \\ \dot{\mathbf x} \end{bmatrix} &=
H\,\mathbf F^1(\mathbf x, t) + H^2\,\mathbf F^{[2]}(\mathbf x, t) =
H \begin{bmatrix} 0 \\ \mathbf f^1(\mathbf x, t) \end{bmatrix} + H^2 \begin{bmatrix} f^{[2]}(\mathbf x, t) \\ \mathbf 0 \end{bmatrix}
\end{align}
where $\mathbf f^1(\mathbf x, t)$ is given by the square brackets in~\eqref{E: Sprime}, \eqref{E: Oprime} and~\eqref{E: phiprime}, and $f^{[2]}(\mathbf x, t)$ by the square bracket in~\eqref{E: Raychaudhuri}. Note that $\mathbf f^1, f^{[2]}$ are independent of $H$. We see that~\eqref{E: quasi-standard form} is resembling the standard form~\eqref{E: standard form} with $H(t)$ playing the role of the perturbation parameter $\epsilon$. The crucial difference is however, that $H$ is time dependent and itself subject to an evolution equation which is part of the system. Because of this we say that~\eqref{E: quasi-standard form} has \emph{quasi-standard form}.

\subsection{Averaging conjecture}\label{SS: averaging conjecture}

Because~\eqref{E: quasi-standard form} has only quasi-standard form and not standard form, we cannot directly apply lemmas~\ref{L: averaging} or~\ref{L: Eckhaus} to obtain rigorous error estimates. However we formulate the following conjecture:
\begin{conjecture}\label{CJ: 1}
Consider some arbitrary non-vacuum ($\Omega>0$) initial value satisfying the Hamiltonian constraint~\eqref{E: hc} and $H>0$. Let $\mathbf [H(t), \mathbf x(t)]^\mathrm{T}$ denote the respective solution to~\eqref{E: quasi-standard form} and let $\mathbf z(t)$ denote the solution to the corresponding averaged equation $\dot{\mathbf z} = H(t)\,\overline{\mathbf f}^1(\mathbf z)$, with $\overline{\mathbf f}^1$ as in~\eqref{E: f bar}. Let $\mathbf X(t), \mathbf Z(t)$ denote the 2-vectors containing the $\Sigma_+$ and $\Omega$ components of the corresponding 3-vectors $\mathbf x(t), \mathbf z(t)$. Then there exists a $t_*$ such that
\begin{align}\label{E: conjecture}
\mathbf X(t)-\mathbf Z(t) &= \mathcal O\big(H(t)\big) \quad\forall\,t>t_*.
\end{align}
\end{conjecture}
In \textsc{Section}~\ref{S: analytical support} we provide analytical support for this conjecture and \textsc{Section}~\ref{S: numerical support} also contains numerical support. Under the premise that it holds, we then derive the future asymptotics of Bianchi~III Einstein-Klein-Gordon cosmologies in \textsc{Section}~\ref{S: asymptotics}. The following observations are key:
\begin{lemma}\label{L: H to 0}
For an initial value with positive $H$, $H$ is strictly monotonically decreasing with $t$, and $\lim_{t\to\infty} H(t)=0.$
\end{lemma}
\begin{proof}
The sign of the Raychaudhuri equation~\eqref{E: Raychaudhuri} is determined by the sign of the factor in square brackets. In \textsc{Subsection}~\ref{SS: basic system} we discussed that $\Sigma_+\in(-1,1)$ and $\Omega\in(0,1)$. With this we see from~\eqref{E: r to O} that $-(1+q)<0$. From~\eqref{E: Raychaudhuri} we also see that $\dot H=0$ for $H=0$. Together this gives the statement of the lemma.
\end{proof}
\begin{proposition}\label{P: error limit}
Validity of \textsc{Conjecture}~\ref{CJ: 1} implies $\lim_{t\to\infty}(\mathbf X(t)-\mathbf Z(t))=0$.
\end{proposition}
\begin{proof}
Cf \textsc{Lemma}~\ref{L: H to 0} together with~\eqref{E: conjecture}.
\end{proof}
Hence, assuming that \textsc{Conjecture}~\ref{CJ: 1} holds, \textsc{Proposition}~\ref{P: error limit} states that the error between $\mathbf X(t)$ and $\mathbf Z(t)$ goes to zero as $t\to\infty$. Consequently the future asymptotics of the averaged system coincides with that of the full system with respect to $\Sigma_+$ and $\Omega$ in this limit. As we show in \textsc{Section}~\ref{S: asymptotics}, modulo a calculation of the rate of approximation this suffices to derive the asymptotic form of the metric. In \textsc{Section}~\ref{S: numerical support} we find agreement of these results with numerics, and also give convincing numerical support for \textsc{Conjecture}~\ref{CJ: 1}. Before all that however, we give strong analytical support for \textsc{Conjecture}~\ref{CJ: 1} in the following section.

\section{Analytical support for conjecture~\ref{CJ: 1}}\label{S: analytical support}

In this section we provide analytical support for \textsc{Conjecture}~\ref{CJ: 1}. In \textsc{Subsection}~\ref{SS: 1st O approx} we argue that, because of \textsc{Lemma}~\ref{L: H to 0}, for large $t$ we can truncate terms of second order in $H$ and work with the resulting first order system in good approximation. We formulate \textsc{Assumption}~\ref{AS: 1} under which this holds. We then perform a qualitative analysis of the corresponding averaged system and identify its attractors in \textsc{Subsection}~\ref{SS: averaged system}; cf \textsc{Lemma}~\ref{L: attractors}. In \textsc{Subsection}~\ref{SS: 1st O errors} we then use this together with lemmas~\ref{L: averaging} and~\ref{L: Eckhaus} to give error estimates between the first order and the averaged system. Finally, in \textsc{Subsecton}~\ref{SS: closing argument} we close our line of arguments heuristically by stating that an infinite series of such first order approximations should yield \textsc{Conjecture}~\ref{CJ: 1} in the continuum limit.

\subsection{The first order approximation}\label{SS: 1st O approx}

As stated in \textsc{Subsection}~\ref{SS: basic system} we consider $H$ to be positive initially. Hence, from \textsc{Lemma}~\ref{L: H to 0} we know that $H(t)$ becomes arbitrarily small as $t$ becomes arbitrarily large. This suggests that for sufficiently large $t$, \eqref{E: quasi-standard form} may be viewed as a perturbative series in $H$ around $H=0$ in analogy to~\eqref{E: standard form}. Hence for sufficiently large $t=t_*$ it seems natural to truncate the second order term to yield the first order approximation\footnote{In~\eqref{E: 1st O approx} $[\mathcal H, \mathbf y]^\mathrm{T}$ represents the same quantities as $[H, \mathbf x]^\mathrm{T}$ in~\eqref{E: quasi-standard form}. We rename the vector to indicate the correspondence of the solutions to the respective systems.}
\begin{align}\label{E: 1st O approx}
\begin{bmatrix}\dot{\mathcal H} \\ \dot{\mathbf y} \end{bmatrix} &=
\begin{bmatrix} 0 \\ \mathcal H\,\mathbf f^1(\mathbf y,t) \end{bmatrix} =
\begin{bmatrix} 0 \\ H_*\,\mathbf f^1(\mathbf y,t) \end{bmatrix}
\end{align}
where $H_*$ denotes the value $H(t_*) = \mathcal H(t_*)$ which is constant due to the first equation.

This step is heuristic since we did not provide proof that we can truncate in good approximation, and we also did not specify what exactly we mean by the latter. We thus assume the following:
\begin{assumption}\label{AS: 1}
The error between $\mathbf x(t)$ and $\mathbf y(t)$ resulting from truncating~\eqref{E: quasi-standard form} to~\eqref{E: 1st O approx} is $\mathcal O(H_*)$ for $t>t_*$.
\end{assumption}
\begin{remark}
\textsc{Section}~\ref{S: numerical support} includes numerical support for the validity of \textsc{Assumption}~\ref{AS: 1}.
\end{remark}

\subsection{Qualitative analysis of the averaged system}\label{SS: averaged system}

The second equation of~\eqref{E: 1st O approx} has standard form~\eqref{E: standard form}. The associated averaged equation~\eqref{E: averaged system} is given by
\begin{align}\label{E: reduced system}
\begin{bmatrix} \dot{\overline\Sigma}_+ \\ \dot{\overline\Omega} \end{bmatrix} &=
H_*\begin{bmatrix}
-\Big(2\big(1-\overline\Sigma_+^2\big) - \frac{\overline\Omega}{2}\Big)\overline\Sigma_+ + 1-\overline\Sigma_+^2-\overline\Omega \\
\overline\Omega\Big( 4\overline\Sigma_+^2 - \big(1-\overline\Omega\big)\Big)
\end{bmatrix}
\end{align}
together with $\dot{\overline\varphi}=0$. Since the latter equation is decoupled, we focus on the reduced system~\eqref{E: reduced system}. By the Hamiltonian constraint~\eqref{E: hc} and the positivity of $\Omega$, the state-space is given by the bounded region
\begin{align*}
\mathcal X := \Big\{ \big(\overline\Sigma_+,\overline\Omega\big)\in\R^2 \Big| \overline\Omega<1-\overline\Sigma_+^2, \overline\Omega>0\Big\}
\end{align*}
depicted in \textsc{Figure}~\ref{F: qualitative flow}.
\begin{figure}
\tdplotsetmaincoords{90}{0}
\begin{tikzpicture}[scale=3.3,tdplot_main_coords]

	\draw[|-|] (-1,0,-.2) -- (0,0,-.2); \draw [-|] (0,0,-.2) -- (1,0,-.2);
	\node at (-1,0,-.3) {$\scriptstyle -1$};
	\node at (0,0,-.3) {$\scriptstyle 0$}; \node at (0,0,-.1) {$\scriptstyle\overline\Sigma_+$};
	\node at (1,0,-.3) {$\scriptstyle 1$};
	\draw[|-|] (-1.2,0,0) -- (-1.2,0,.5); \draw[-|] (-1.2,0,.5) -- (-1.2,0,1);
	\node at (-1.3,0,1) {$\scriptstyle1$};
	\node at (-1.31,0,.5) {$\scriptstyle1/2$};\node at (-1.1,0,.5) {$\scriptstyle\overline\Omega$};
	\node at (-1.3,0,0) {$\scriptstyle0$};
	
	\draw[->, thick] (-1,0,0) ..controls+(.1,0,.2) and+(-.1,0,-.1).. (-1/2,0,3/4);
	\draw[thick] (-1/2,0,3/4) ..controls+(.1,0,.1) and+(-.22,0,0).. (0,0,1);
	\draw[->, thick] (1,0,0) ..controls+(-.1,0,.2) and+(.1,0,-.1).. (1/2,0,3/4);
	\draw[thick] (1/2,0,3/4) ..controls+(-.1,0,.1) and+(.22,0,0).. (0,0,1);
	\draw[->,thick] (-1,0,0) -- (-.25,0,0);\draw[thick] (-.25,0,0) -- (1/2,0,0);
	\draw[->,thick] (1,0,0) -- (.75,0,0);\draw[thick] (.75,0,0) -- (1/2,0,0);
	\draw[->,dashed] (0,0,1) ..controls+(.1,0,-.25) and+(-.1,0,.2).. (.22,0,1/2);
	\draw[dashed] (.22,0,1/2) ..controls+(.1,0,-.2) and+(-.1,0,.2).. (1/2,0,0);
	\draw[->] (-1,0,0) ..controls+(.1,0,.2) and+(-.25,0,0).. (-.185,0,.864);
	\draw (-.185,0,.864) ..controls+(.2,0,0) and+(-.19,0,.3).. (1/2,0,0);
	\draw[->] (-1,0,0) ..controls+(.1,0,.2) and+(-.3,0,0).. (-.33,0,.5);
	\draw (-.33,0,.5,0) ..controls+(.4,0,0) and+(-.19,0,.3).. (1/2,0,0);
	\draw[->] (-1,0,0) ..controls+(.1,0,.2) and+(-.1,0,0).. (-.47,0,.26);
	\draw (-.47,0,.26) ..controls+(.1,0,0) and+(-.19,0,.3).. (1/2,0,0);
	\draw[->] (1,0,0) ..controls+(-.1,0,.2) and+(.25,0,0).. (.185,0,.864);
	\draw (.185,0,.864) ..controls+(-.18,0,0) and+(-.15,0,.3).. (1/2,0,0);
	\draw[->] (1,0,0) ..controls+(-.1,0,.2) and+(.18,0,0).. (.4,0,.5);
	\draw (.4,0,.5) ..controls+(-.13,0,0) and+(-.15,0,.3).. (1/2,0,0);
	\draw[->] (1,0,0) ..controls+(-.1,0,.2) and+(.05,0,0).. (.6,0,.26);
	\draw (.6,0,.26) ..controls+(-.03,0,0) and+(-.15,0,.3).. (1/2,0,0);
	
	\draw plot [mark=*, mark size=.5] coordinates{(0,0,1)};
	\node at (0,0,1.1) {$\scriptstyle F$};
	\draw plot [mark=*, mark size=.5] coordinates{(-1,0,0)};
	\node at (-1.1,0,-.1) {$\scriptstyle T$};
	\draw plot [mark=*, mark size=.5] coordinates{(1,0,0)};
	\node at (1.1,0,-.1) {$\scriptstyle Q$};
	\draw plot [mark=*, mark size=.5] coordinates{(1/2,0,0)};
	\node at (1/2,0,-.1) {$\scriptstyle D$};
	
\end{tikzpicture}
\caption{The flow of the reduced averaged system~\eqref{E: reduced system} on $\mathrm{closure}(\mathcal X)$. Solutions inside $\mathcal X$ are LRS Bianchi~III matter solutions ($\overline\Omega>0$). $\overline\Omega=0$ marks the vacuum boundary. Solutions on the $\overline\Omega=1-\overline\Sigma_+^2$ boundary are of LRS Bianchi type~I. The exact solutions associated with the equilibrium points are listed in \textsc{Table}~\ref{Tbl: 1}. We point out the analogy between this diagram and the corresponding diagram for perfect fluid solutions \cite[p~202, \sc Fig~9.1]{WainwrightEllis1997}.}
\label{F: qualitative flow}
\end{figure}
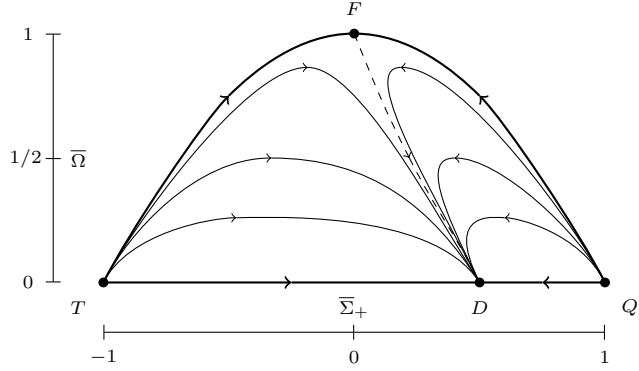 
The system~\eqref{E: reduced system} admits the four equilibrium points $T=(-1,0)$, $Q=(1,0)$, $D=(1/2,0)$ and $F=(0,1)$. To formulate the following lemma we define the LRS Bianchi~I state-space
\begin{align}\notag
\mathcal X_I := \Big\{ \big(\overline\Sigma_+,\overline\Omega\big)\in\R^2 \Big| \overline\Omega=1-\overline\Sigma_+^2, \overline\Omega>0 \Big\} .
\end{align}
\begin{lemma}\label{L: attractors}
$D$ is the future attractor for the flow of~\eqref{E: reduced system} in $\mathrm{closure}(\mathcal X)\setminus\mathcal{\overline X}_I$. $F$ is the future attractor for the flow in $\mathcal X_I$. There is a separatrix from $F$ to $D$. $T$ ($Q$) is the past attractor for solutions in $\mathrm{closure}(\mathcal X)$ with initial data left (right) of the separatrix.
\end{lemma}
\begin{proof}
 A check of the eigenvalues of the linearisation of~\eqref{E: reduced system} at these points reveals that $T, Q$ and $F$ are hyperbolic. In particular $T$ and $Q$ are local sources while $F$ is a local saddle with its stable manifold tangential to $\overline\Sigma_+$ direction. $D$ is degenerate with a stable manifold tangent to $\overline\Sigma_+$ direction and a centre manifold tangent to the vector $[-1/2,1]^\mathrm{T}$. We prove in \textsc{Lemma}~\ref{L: stability of D} that the flow on that arm of the centre manifold which intersects the state-space is stable. Hence $D$ acts as a local sink for the flow in the closure of the state-space. The information gathered suffices to draw the qualitative flow diagram of \textsc{Figure}~\ref{F: qualitative flow} from which the statements of the lemma are apparent.
\end{proof}

\subsection{Error estimates for the first order approximation}\label{SS: 1st O errors}

For some initial value, let $\mathbf y(t)$ denote the solution of the first order approximation~\eqref{E: 1st O approx} and let $\mathbf z(t)$ denote the solution of the corresponding averaged system given by \eqref{E: reduced system} together with $\dot{\overline\varphi}=0$. Then from \textsc{Lemma}~\ref{L: averaging} we know that $\mathbf y(t) - \mathbf z(t) = \mathcal O(H_*)$ on time scales of $\mathcal O(H_*^{-1})$. Furthermore, by \textsc{Lemma}~\ref{L: attractors} we have a case of averaging with attraction and, in analogy to the discussion of \textsc{Subsection}~\ref{SS: Van der Pol}, we can extend the validity of this error estimate for all times for the $\Sigma_+$ and $\Omega$ components. Hence, denoting the respective 2-vectors by capital letters we have
\begin{align}\label{E: YZ estimate}
\mathbf Y(t) - \mathbf Z(t) &=
\begin{bmatrix} \Sigma_+(t) \\ \Omega(t) \end{bmatrix} -
\begin{bmatrix} \overline\Sigma_+(t) \\ \overline\Omega(t) \end{bmatrix} =
\mathcal O(H_*) .
\end{align}

\subsection{Closing the argument}\label{SS: closing argument}

The steps taken so far are as follows:
\begin{enumerate}[(a)]
\item We considered the LRS Bianchi~III Einstein-Klein-Gordon system in quasi-standard form~\eqref{E: quasi-standard form} with solution $[H(t),\mathbf x(t)]^\mathrm{T}$, and for which we showed in \textsc{Lemma}~\ref{L: H to 0} that $H(t)\to0$ as $t\to0$. \label{I: a}
\item We thus argued that for sufficiently large $t=t_*$ we should be able to truncate and work with the first order system~\eqref{E: 1st O approx} with solution $[\mathcal H(t), \mathbf y(t)]^\mathrm{T}$ in good approximation.\label{I: b}
\item Since the previous argument is heuristic, we made \textsc{Assumption}~\ref{AS: 1} under which it is valid. That is, we assumed that the error from truncation between $\mathbf x(t)$ and $\mathbf y(t)$ is $\mathcal O\big(H_*)$ for $t>t_*$. \label{I: c}
\item In this approximation $\dot{\mathcal H}=0$. Consequently its reduced system reads $\dot{\mathbf y}=H_*\mathbf f^1(\mathbf y, t)$ with $H_*=H(t_*)=\mathcal H(t_*)$ and has standard form~\eqref{E: standard form}. \label{I: d}
\item We applied \textsc{Lemma}~\ref{L: averaging} to obtain an $\mathcal O(H_*)$ error estimate between $\mathbf y(t)$ and the solution $\mathbf z(t)$ of the corresponding averaged system $[\eqref{E: reduced system}, \dot\varphi=0]$ on timescales of $\mathcal O(H_*^{-1})$. \label{I: e}
\item By \textsc{Lemma}~\ref{L: Eckhaus} we extended this estimate for the $\Sigma_+, \Omega$ components, which we denoted by the 2-vectors $\mathbf Y(t),\mathbf Z(t)$, to all times; cf~\eqref{E: YZ estimate}. \label{I: f}
\end{enumerate}
Hence, we have an $\mathcal O(H_*)$ truncation error estimate between $\mathbf x(t)$ and $\mathbf y(t)$ for $t>t_*$ from \textsc{Assumption}~\ref{AS: 1} in step~(\ref{I: c}). Furthermore we have an $\mathcal O(H_*)$ averaging error estimate between $\mathbf Y(t)$ and $\mathbf Z(t)$ for $t>t_*$ from step~(\ref{I: f}). In combination we have
\begin{align}\label{E: error estimate}
\mathbf X(t) - \mathbf Z(t) = \mathcal O(H_*) \quad\forall\,t>t_*
\end{align}
where $\mathbf X(t)$ denotes the 2-vector containing the $\Sigma_+$ and $\Omega$ components of $\mathbf x(t)$.

Now recall that $H_*$ denotes the value $H(t)$ at the time of truncation $t=t_*$ in~\eqref{I: c}. However because of \textsc{Lemma}~\ref{L: H to 0} in~\eqref{I: a}, we can choose $t_*$ arbitrarily large and $H_*$ arbitrarily small. In the continuum limit of infinitely many such truncation times this should yield~\eqref{E: conjecture} and thus \textsc{Conjecture}~\ref{CJ: 1}. We support this statement with the following line of arguments:

Let $t_*$ be an arbitrary truncation time in~\eqref{I: c} which yields the first order approximation \eqref{E: 1st O approx} with parameter $H_*=H(t_*)$. Consider now the infinite sequence $(t_k)_{k=0}^\infty$, with $t_k:=t_*+k\epsilon$ and $\epsilon>0$. Clearly we have $t_{k+1}>t_k$ for all $k$. Hence because of \textsc{Lemma}~\ref{L: H to 0} we can consider each of them as separate truncation times, yielding the infinite sequence of first order approximations with parameters $(H_k=H(t_k))_{k=0}^\infty$. For each we can then make \textsc{Assumption}~\ref{AS: 1} and perform steps~\eqref{I: c}--\eqref{I: f} to yield the infinite sequence of error estimates $(\mathbf X(t) - \mathbf Z(t) = \mathcal O(H_k))_{k=0}^\infty$; cf~\eqref{E: error estimate}. In the continuum limit $(t_k)_{k=0}^\infty\to[t_*,\infty)$ we should have $(H_k)_{k=0}^\infty\to H(t)$ and $(\mathbf X(t) - \mathbf Z(t) = \mathcal O(H_k))_{k=0}^\infty\to\eqref{E: conjecture}$, which is the statement of \textsc{Conjecture}~\ref{CJ: 1}.

\begin{remark}
For the scope of the present work, we refrain from giving rigorous proofs of \textsc{Assumption}~\ref{AS: 1} and of the continuum limit above. Because of this, the content of this section does not qualify as proof of \textsc{Conjecture}~\ref{CJ: 1} but rather as strong analytical support of it. \textsc{Section}~\ref{S: numerical support} contains further numerical support. For our analysis on the LRS Bianchi~III future asymptotics in \textsc{Section}~\ref{S: asymptotics} below we work under the premise that the conjecture holds.
\end{remark}

\section{The future asymptotics of LRS Bianchi~III Einstein-Klein-Gordon cosmologies}\label{S: asymptotics}

In the previous section we provided analytical support for \textsc{Conjecture}~\ref{CJ: 1}. Under the premise that it holds, in this section we derive the future asymptotics of LRS Bianchi~III Einstein-Klein-Gordon cosmologies. In \textsc{Subsection}~\ref{SS: exact solutions} we give the exact solutions associated with the equilibrium points of the averaged system. In particular we show that $D$ can be associated with the Bianchi~III form of flat spacetime. We then interpret the latter in the context of LRS Bianchi~III matter solutions in \textsc{Subsection}~\ref{SS: interpretation of eqp sols}, and argue that in order to determine the asymptotic form of the metric we also have to find the rate of approximation of solutions to $D$. We achieve this by means of a centre manifold analysis in \textsc{Subsection}~\ref{SS: CM analysis of D}; cf \textsc{Lemma}~\ref{L: SO(tau) asympt}. After that, we present our main results in \textsc{Subsection}~\ref{SS: BIII future asympt} which give the dominant future asymptotic behaviour of LRS Bianchi~III Einstein-Klein-Gordon solutions; cf theorems~\ref{T: X asympt} and~\ref{T: g asympt}. Finally, in \textsc{Subsection}~\ref{SS: vacuum case} we compare this behaviour to the corresponding future asymptotics in the vacuum case. We point out the quantitative difference and thus speak of the Klein-Gordon future asymptotics to be matter dominated. At the end of this section we also take this opportunity to shed some light on the question when the rate of approximation to the fixed point has to be taken into account and when not.

\subsection{Exact solutions associated with the equilibrium points}\label{SS: exact solutions}

Each equilibrium point of the averaged system~\eqref{E: reduced system} can be associated with an exact solution of the Einstein equations. This is true since for equilibrium solutions the values of $\Sigma_+$ and $\Omega$ are by definition constant and hence the evolution equations become particularly simple and can be solved exactly. The results are summarised in~\textsc{Table}~\ref{Tbl: 1}. In the following we lay out the calculation at the example of $D$.
\begin{table}
\begin{tabular}{ l | l | l | l }
  equil. point 		& $a(t)$		& $b(t)$ 		& solution 						\\ \hline
  $T$				& $c_1t$		& $c_2$		& Taub (flat LRS Kasner)			\\
  $Q$			& $c_1t^{-1/3}$	& $c_2t^{2/3}$	& non-flat LRS Kasner			\\
  $D$			& $c_1$		& $c_2t$		& Bianchi~III form of flat spacetime	\\
  $F$				& $c_1t^{2/3}$	& $c_2t^{2/3}$	& Einstein-de-Sitter (flat Friedman)
\end{tabular} \vspace{.4 cm}
\caption{The exact solutions associated with the equilibrium points of the reduced averaged system~\eqref{E: reduced system}. $a(t), b(t)$ denote the scale factors of the metric. $c_1,c_2\in\R$.}
\label{Tbl: 1}
\end{table}

Starting with the Raychaudhuri equation~\eqref{E: Raychaudhuri} and evaluating it at $D=(1/2,0)$ gives
\begin{align}\label{E: H at D}
\dot H(t) = -\frac{3}{2}H(t)^2 \quad\Longrightarrow\quad H(t) = \frac{2}{3t} \quad\text{for large $t$}.
\end{align}
From the evolution equations for the spatial metric components (cf for instance~\cite[(2.29)]{Rendall2008}) and the variable definitions of $H$ and $\Sigma_+$ given in \textsc{Subsection}~\ref{SS: basic system} we have
\begin{align}\label{E: a dot}
\dot a &= a H (1-2\Sigma_+) \quad\text{and}\quad \dot b = b H (1+\Sigma_+)
\end{align}
for the evolution equations for the scale factors of the LRS Bianchi~III metric~\eqref{E: BIII metric}. Evaluating these at $D$ yields $\dot a = 0$ and $\dot b = 3bH(t)/2$, and using~\eqref{E: H at D} we get
\begin{align}\label{E: D metric}
a(t)=c_1 \quad\text{and}\quad b(t)=c_2 t
\end{align}
where here and in the following $c_1,c_2\in\R_+$. This is the \emph{Bianchi~III form of flat spacetime}; cf~\cite[p~193, (9.7)]{WainwrightEllis1997}.

An analogous calculation yields the scale factors for the solutions associated with $T$ and $Q$. For $T$ we get $a(t) = c_1 t$ and $b(t) = c_2$, which is the \emph{Taub Kasner} solution; cf~\cite[\textsc{Sec}~6.2.2 and p~193, (9.6)]{WainwrightEllis1997}. It is spatially flat and of Bianchi type~I. For $Q$ we get $a(t) = c_1 t^{-1/3}$ and $b(t) = c_2 t^{2/3}$, which is the \emph{non-flat LRS Kasner} Bianchi~I solution; cf~\cite[\textsc{Sec}~6.2.2 and \textsc{Sec}~9.1.1 (2)]{WainwrightEllis1997}. Both are vacuum solutions in the sense that $\Omega=0$.

From the interpretation of $\Sigma_+$ as the shear parameter given in \textsc{Subsection}~\ref{SS: basic system} we can readily infer that $F$ corresponds to an isotropically expanding solution since it has vanishing shear. But we also calculate the asymptotic rate of expansion for large $t$. The Raychaudhuri equation~\eqref{E: Raychaudhuri} evaluated at $F$ reads
\begin{align}\label{E: H at F}
\dot H(t) = -3H(t)^2 \cos(t-\varphi(t)).
\end{align}
Recalling that in the averaged system $\dot{\overline\varphi}=0$ (cf \textsc{Subsection}~\ref{SS: averaged system}) we choose $\varphi(t)=\overline\varphi=0$ and solve~\eqref{E: H at F} with the initial condition $H(0)\to\infty$ to obtain\footnote{For a general $\overline\varphi$ the denominator contains a more complicated sum of $\sin$ and $\cos$ product terms. For large $t$ these are however also dominated by the $3t$ term. Hence, with respect to the asymptotic form of $H(t)$, our choice $\overline\varphi=0$ is without loss of generality. }
\begin{align}\label{E: H at F}
H(t) &= \frac{2}{3\sin(t)\cos(t)+3t}\approx\frac{2}{3t} \text{ for large }t.
\end{align}
Evaluating~\eqref{E: a dot} with~\eqref{E: H at F} and $\Sigma_+=0$ yields $a(t) = c_1 t^{2/3}$ and $b(t) = c_2 t^{2/3}$. Hence for large $t$ and $c_1=c_2$ the equilibrium point $F$ can be associated with the flat Friedman \emph{Einstein-de-Sitter} solution; cf~\cite[\sc Sec~9.1.1 (1)]{WainwrightEllis1997} with $\gamma=2$.

Concerning the matter field we can take the Klein-Gordon equation in its original form~\eqref{E: KG} and plug in the asymptotic behaviour of $H$, in order to find the asymptotic behaviour of $\phi$ for the respective solutions. For instance, for $D$ and $F$ this way we find that $\phi(t)\propto t^{-1}\sin t$ approximately for large $t$. Alternatively we could work with the two first order equations~\eqref{E: Oprime}--\eqref{E: phiprime} together with the transformation~\eqref{E: amplitude-phase} to obtain this result.

\subsection{Interpretation of $D$ and the Bianchi~III form of flat spacetime in the context of LRS Bianchi~III non-vacuum solutions}\label{SS: interpretation of eqp sols}

In \textsc{Lemma}~\ref{L: attractors} we showed that $D$ is the future attractor for solutions in $\mathcal X$ of the reduced averaged system~\eqref{E: reduced system}, ie for LRS Bianchi~III non-vacuum solutions. Assuming that \textsc{Conjecture}~\ref{CJ: 1} holds this implies that $\lim_{t\to\infty}\mathbf X(t)=D$ as well, ie that the $\Sigma_+$ and $\Omega$ coordinates of non-vacuum solutions of the original system~\eqref{E: quasi-standard form} are asymptotic to the same values.

In \textsc{Subsection}~\ref{SS: exact solutions} we associated  with $D$ the Bianchi~III form of flat spacetime which is characterised by the scale factors~\eqref{E: D metric}. We can however \emph{not} infer from this, that the scale factors for solutions in $\mathcal X$ approach~\eqref{E: D metric} in the limit. This is because $D$ fixes the limit form of~\eqref{E: a dot} and thus only the \emph{time derivative} of the scale factors. Hence the latter will coincide with that of the Bianchi~III form of flat spacetime in the limit, but not necessarily the scale factors themselves. We can however obtain the asymptotic form of the scale factors by taking into account the \emph{rate of approximation} of solutions in $\mathcal X$ to $D$.

As stated in the proof of \textsc{Lemma}~\ref{L: attractors}, $D$ is degenerate with a stable manifold tangent to $\overline\Sigma_+$ direction and a centre manifold tangent to the vector $[-1/2,1]^\mathrm{T}$. We can thus get the stability of $D$ for solutions in $\mathcal X$ as well as the approximation rate for these solutions from a \emph{centre manifold analysis} (cf \textsc{Appendix}~\ref{A: CM analysis}) as follows.

\subsection{Centre manifold analysis of $D$}\label{SS: CM analysis of D}

We start by relating~\eqref{E: reduced system} to a system of the form~\eqref{E: CM x dot}--\eqref{E: CM y dot}. Firstly, 
let us go back for a moment to step~\eqref{I: a} of the line of arguments of \textsc{Subsection}~\ref{SS: closing argument}, and hence to the original equation in quasi-standard form~\eqref{E: quasi-standard form}. We transform to a \emph{rescaled time} $\tau$ according to
\begin{align}\label{E: t to tau}
t&\mapsto\tau:=\int_{t_0}^t H(s)\,\mathrm ds \quad\Longrightarrow\quad
\partial_t=H\partial_\tau .
\end{align}
For convenience we denote the time transformed quantities by the same letters as before, but view them now as functions of $\tau$ instead of $t$. Hence we write the resulting system as
\begin{align}\label{E: guiding quasi-standard}
\begin{bmatrix} \partial_\tau H \\ \partial_\tau\mathbf x \end{bmatrix} &=
\mathbf F^1(\mathbf x, \tau) + H\,\mathbf F^{[2]}(\mathbf x, \tau)
\end{align}
and denote its solution by $[H(\tau), \mathbf x(\tau)]^\mathrm{T}$. Applying now the analogous steps to~\eqref{I: b}--\eqref{I: d} of \textsc{Subsection}~\ref{SS: closing argument}, we obtain the reduced averaged system
\begin{align}\label{E: guiding system}
\begin{bmatrix} \partial_\tau\overline\Sigma_+ \\ \partial_\tau\overline\Omega \end{bmatrix} &=
\begin{bmatrix}
-\Big(2\big(1-\overline\Sigma_+^2\big) - \frac{\overline\Omega}{2}\Big)\overline\Sigma_+ + 1-\overline\Sigma_+^2-\overline\Omega \\
\overline\Omega\Big( 4\overline\Sigma_+^2 - \big(1-\overline\Omega\big)\Big)
\end{bmatrix}
\end{align}
where again all quantities are functions of $\tau$.
\begin{remark}\label{R: guiding system}
\eqref{E: guiding system} is the \emph{guiding system} of~\eqref{E: reduced system}; cf \cite[p~30]{SandersEtAl2007}. Since their flows only differ by a time rescaling the corresponding flow plots are equivalent. Because of the analogous relation between~\eqref{E: quasi-standard form} and~\eqref{E: guiding quasi-standard} we also say that the latter is the guiding system of the former.
\end{remark}
Secondly, we transform $\overline\Sigma_+\mapsto S:=\overline\Sigma_+-1/2$ to shift the critical point $D$ to the origin. Finally, we diagonalise the system using the eigenvectors $[1,0]^\mathrm{T}, [-1/2,1]^\mathrm{T}$ of its linearisation at the origin: Defining $x,y$ by the linear transformation
\begin{align}\label{E: SO to xy}
\begin{bmatrix} S \\ \overline\Omega \end{bmatrix} &=
y \begin{bmatrix} 1 \\ 0 \end{bmatrix} +
x \begin{bmatrix} -1/2 \\ 1 \end{bmatrix} =
\begin{bmatrix} y-x/2 \\ x \end{bmatrix}
\end{align}
we obtain the new system
\begin{align}\label{E: xy system}
\begin{bmatrix} \partial_\tau x \\ \partial_\tau y \end{bmatrix} &=
\begin{bmatrix} Ax + f(x,y) \\ By + g(x,y) \end{bmatrix} , \quad
\end{align}
with $A=0, B=-3/2$,
\begin{align}
f(x,y) &= x^3-4x^2y+4xy^2-x^2+4xy , \label{E: f} \\
g(x,y) &= 2y^3-y^2x-\frac{1}{2}yx^2+\frac{1}{4}x^3+2y^2+\frac{1}{2}yx-\frac{1}{4}x^2 . \notag
\end{align}
\begin{remark}\label{R: topological equivalence}
Since \eqref{E: SO to xy} acts as a homeomorphism between~\eqref{E: guiding system} and~\eqref{E: xy system}, the respective flows are topologically equivalent. Because of \textsc{Remark}~\ref{R: guiding system} the same is true for the flow of~\eqref{E: reduced system}. Note also that the linearisation of~\eqref{E: xy system} at the origin has eigenvectors $[1,0]^\mathrm{T}, [0,1]^\mathrm{T}$ to the same eigenvalues $0, -3/2$ as the linearisation of~\eqref{E: guiding system} at $D$, (just in opposite order).
\end{remark}

\eqref{E: xy system} is of the form~\eqref{E: CM x dot}--\eqref{E: CM y dot}. Hence, from \textsc{Lemma}~\ref{T: CM existence} we know that there exists a local centre manifold $y=h(x)$ tangent to the origin. From the eigenvector to the eigenvalue 0 we know that it is tangent to the $x$-direction. As described at the end of \textsc{Appendix}~\ref{A: CM analysis}, we can approximate the centre manifold to arbitrary degree. For our purposes a lowest non-vanishing order approximation is sufficient. Choosing $\phi(x)=-\frac{1}{6}x^2$ and applying the function $M$ defined by~\eqref{E: M map}, we have $(M\phi)(x)=\mathcal O(x^3)$ as $x\to0$. By \textsc{Lemma}~\ref{T: CM approx} this implies that
\begin{align}\label{E: h approx}
h(x) &= -\frac{1}{6}x^2+\mathcal O(x^3) \quad\text{as}\quad x\to0.
\end{align}
The flow on the centre manifold is governed by~\eqref{E: flow on CM} which in our case reduces to $\partial_\tau u = f(u,h(u))$. Plugging in~\eqref{E: f} and~\eqref{E: h approx} yields
\begin{align}\label{E: u prime}
\partial_\tau u = -u^2+\mathcal O(u^3) \quad\text{as}\quad
u\to0,
\end{align}
and by \textsc{Lemma}~\ref{T: CM stability and flow} the solution to this equation governs the local stability of the origin of~\eqref{E: xy system} and the local flow around it.

\begin{lemma}\label{L: stability of D}
$D$ is locally asymptotically stable for the flow of~\eqref{E: reduced system} in $\mathcal X$.
\end{lemma}
\begin{proof}

The zero solution of \eqref{E: u prime} is attracting solutions with positive $u$. From \textsc{Lemma}~\ref{T: CM stability and flow}\eqref{TI: b} we then get that the zero solution of~\eqref{E: xy system} is locally asymptotically stable for the flow with $x\geq0$. The claim of the lemma then follows from the statement of \textsc{Remark}~\ref{R: topological equivalence}, noting that $x\geq0$ corresponds to $\overline\Omega\geq0$; cf~\eqref{E: SO to xy}.
\end{proof}
\begin{remark}
\textsc{Lemma}~\ref{L: stability of D} fills the blank in the proof of \textsc{Lemma}~\ref{L: attractors}.
\end{remark}
\begin{remark}
Since there is no attracting structure other than $D$ for the flow in $\mathcal X$, the statement of \textsc{Lemma}~\ref{L: stability of D} actually holds \emph{globally} for that flow, not only locally; cf the proof of \textsc{Lemma}~\ref{L: attractors}.
\end{remark}
\begin{remark}
$D$ is a \emph{saddle-node} and its local flow is topologically equivalent to the diagram in~\cite[p~149, \sc Fig~3]{Perko2001}.
\end{remark}
\begin{lemma}\label{L: SO(tau) asympt}
Solutions of~\eqref{E: guiding system} behave like
\begin{align}\label{E: SO(tau)}
\overline\Sigma_+(\tau)\approx\frac{1}{2}-\frac{1}{2\tau} \quad\text{and}\quad
\overline\Omega(\tau)\approx\frac{1}{\tau} \quad\text{for large $\tau$}.
\end{align}
\end{lemma}
\begin{proof}
Solving~\eqref{E: u prime} to lowest order yields $u(\tau)\approx 1/\tau$ for large $\tau$. By \textsc{Lemma}~\ref{T: CM stability and flow}\eqref{TI: b} and using~\eqref{E: h approx} we then have $x(\tau)\approx 1/\tau$ and $y(\tau)\approx -\frac{1}{6\tau}$ for large $\tau$. Translating back to the original variables using~\eqref{E: SO to xy} concludes the proof.
\end{proof}

\subsection{The future asymptotics of LRS Bianchi~III non-vacuum solutions}\label{SS: BIII future asympt}

Let $\mathbf X(\tau)=[\Sigma_+(\tau),\Omega(\tau)]^\mathrm T$ denote the respective components of the solution to the guiding system~\eqref{E: guiding quasi-standard} of~\eqref{E: quasi-standard form}. Let $\mathbf Z(\tau)=[\overline\Sigma_+(\tau),\overline\Omega(\tau)]^\mathrm T$ denote the solution to the reduced part of the corresponding averaged guiding system~\eqref{E: guiding system}. \textsc{Lemma}~\ref{L: SO(tau) asympt} gives the behaviour of $\mathbf Z(\tau)$ for large $\tau$. Assuming that \textsc{Conjecture}~\ref{CJ: 1} holds, this gives also an estimate
\begin{align}\label{E: X1}
\mathbf X(\tau) = \mathbf Z(\tau)+\mathcal O(H(\tau)) \quad\text{for large $\tau$}.
\end{align}
To quantify this estimate we need to determine the dominant behaviour of $H$ for large times.
\begin{lemma}\label{L: H for large times}
Assuming that \textsc{Conjecture}~\ref{CJ: 1} holds, $H(t) \approx \frac{2}{3t}$ for large $t$.
\end{lemma}
\begin{proof}
We know from \textsc{Lemma}~\ref{L: SO(tau) asympt} (or \textsc{Lemma}~\ref{L: attractors}) that $\lim_{t\to\infty}\mathbf Z(t)=D$. Provided that \textsc{Conjecture}~\ref{CJ: 1} holds we then know from \textsc{Proposition}~\ref{P: error limit} that also $\lim_{t\to\infty}\mathbf X(t)=D$. In \textsc{Subsection}~\ref{SS: exact solutions} we established that the exact solution associated with $D$ is the Bianchi~III form of flat spacetime, and we calculated its Hubble scalar in~\eqref{E: H at D}. The Hubble scalar of a solution which approaches $D$ will thus approach the Hubble scalar of the Bianchi~III form of flat spacetime. Since $\mathbf X$ approaches $D$ we thus can infer that the corresponding Hubble scalar has the form $H(t)=\frac{2}{3t}(1+o(1))$. The lowest order approximation concludes the proof.

Alternatively we can follow from \textsc{Lemma}~\ref{L: SO(tau) asympt} and \textsc{Conjecture}~\ref{CJ: 1} that
\begin{align*}
\mathbf X(\tau)=\begin{bmatrix}
\frac{1}{2}-\frac{1}{2\tau}+\mathcal O(H(\tau)) \\
\frac{1}{\tau}+\mathcal O(H(\tau))
\end{bmatrix} \quad\text{for large $\tau$}.
\end{align*}
Plugging this into the Raychaudhuri equation~\eqref{E: Raychaudhuri} we have $\dot H = -H^2(\frac{3}{2}+o(1))$ from which one can show that $H(t)=\frac{2}{3t}(1+o(1))$ follows; cf also the analogous case in~\cite[p~1290]{Rendall2002}.
\end{proof}
From this result we immediately get the asymptotic form of the time transformation~\eqref{E: t to tau}.
\begin{lemma}\label{L: t to tau asympt}
Assuming that \textsc{Conjecture}~\ref{CJ: 1} holds, $\tau(t)\approx\frac{2}{3}\ln t$ for large $t$.
\end{lemma}
\begin{proof}
Use \textsc{Lemma}~\ref{L: H for large times} in~\eqref{E: t to tau}.
\end{proof}

Combining now lemmas~\ref{L: H for large times} and~\ref{L: t to tau asympt} with~\eqref{E: X1} we have
\begin{align}\label{E: X2}
\mathbf X(t) = \mathbf Z(t)+\mathcal O(t^{-1}) =
\begin{bmatrix}
\frac{1}{2}-\frac{3}{4\ln t} \\
\frac{3}{2\ln t}
\end{bmatrix}
+ \mathcal O(t^{-1}) \quad\text{for large $\tau$}.
\end{align}

We are now ready to formulate our main results.

\begin{theorem}\label{T: X asympt}
Assuming that \textsc{Conjecture}~\ref{CJ: 1} holds, the shear parameter $\Sigma_+$ and the rescaled energy density $\Omega$ of LRS Bianchi~III Einstein-Klein-Gordon solutions (cf \textsc{Subsection}~\ref{SS: basic system}) behave like
\begin{align}\notag
\Sigma_+(t)\approx\frac{1}{2}-\frac{3}{4\ln t} \quad\text{and}\quad
\Omega(t)\approx\frac{3}{2\ln t} \quad\text{for large t}.
\end{align}
\end{theorem}
\begin{proof}
This is an immediate consequence of~\eqref{E: X2} since $t^{-1}$ falls off faster than $(\ln t)^{-1}$.
\end{proof}
\begin{remark}
The behaviour of $\Sigma_+,\Omega$ found in \textsc{Theorem}~\ref{T: X asympt} is identical to what was found in~\cite[p~1290]{Rendall2002} in the LRS Bianchi~III Einstein-Vlasov case.
\end{remark}

What is left to do is to translate the above result to an asymptotic behaviour of the metric in terms of metric time. We do this with the following theorem, to which we also add the asymptotic behaviour of the Klein-Gordon field.
\begin{theorem}\label{T: g asympt}
Assuming that \textsc{Conjecture}~\ref{CJ: 1} holds, the scale factors of the LRS Bianchi~III Einstein-Klein-Gordon metric~\eqref{E: BIII metric} behave like
\begin{align}\label{E: g asympt}
a(t)\approx c_1\ln t \quad\text{and}\quad
b(t)\approx c_2 t \quad\text{for large $t$},
\end{align}
and the Klein-Gordon field behaves like
\begin{align}\notag
\phi(t)\approx \frac{c_3}{t}\sin t \quad\text{for large $t$},
\end{align}
with positive constants $c_1, c_2, c_3$.
\end{theorem}
\begin{proof}
The statement for the matter field follows from the Klein-Gordon equation~\eqref{E: KG} together with the asymptotic behaviour of $H$ from \textsc{Lemma}~\ref{L: H for large times}. Alternatively one can use \textsc{Lemma}~\ref{L: H for large times} and \textsc{Theorem}~\ref{T: X asympt} together with the left equation of~\eqref{E: r to O} to find that $r\approx \widetilde c_3 t^{-1}(\ln t)^{-1/2}\approx c_3 t^{-1}$ for large $t$, with positive constants $\widetilde c_3, c_3$. The transformation~\eqref{E: amplitude-phase} then yields the result if we note that from~\eqref{E: phiprime} and \textsc{Lemma}~\ref{L: H for large times} we know that the phase $\varphi$ goes to a constant.)

For the scale factors, transforming $t\mapsto\tau$ according to~\eqref{E: t to tau} their evolution equations read $\partial_\tau a(\tau) = a(\tau)(1-2\Sigma_+(\tau))$ and $\partial_\tau b(\tau) = b(\tau)(1+\Sigma_+(\tau))$. Plugging in the result of \textsc{Theorem}~\ref{T: X asympt} we obtain
\begin{align*}
\partial_\tau a(\tau) &\approx a(\tau)/\tau & &\Longrightarrow &
a(\tau)&\approx c_1\tau \\
\partial_\tau b(\tau) &\approx b(\tau)\left(\frac{3}{2}-\frac{1}{2\tau}\right) & &\Longrightarrow &
b(\tau) &\approx \widetilde{c_2}\frac{\mathrm e^{3\tau/2}}{\sqrt\tau}
\approx \widehat{c_2}\,\mathrm e^{3\tau/2}
\end{align*}
for large $\tau$, with positive constants $c_1,\widetilde{c_2},\widehat{c_2}$. In the last approximation for $b(\tau)$ we acknowledged that for large $\tau$ the $\tau$ dependence is dominated by the exponential such that the square-root can be absorbed into the integration constant in good approximation. Transforming to metric time using~\textsc{Lemma}~\ref{L: t to tau asympt} concludes the proof.
\end{proof}

\subsection{Comparison to vacuum solutions; matter dominance}\label{SS: vacuum case}

The future asymptotics for LRS Bianchi~III vacuum cosmologies is known to be the Bianchi~III form of flat spacetime (cf eg~\cite[\textsc Sec~10.5]{Rendall2008} or \cite{WainwrightEllis1997}) and we wrote down the scale factors of this solution in~\eqref{E: D metric} and \textsc{Table}~\ref{Tbl: 1}. Comparing this to our result~\eqref{E: g asympt} of \textsc{Theorem}~\ref{T: g asympt} we see that the two solutions agree with respect to the scale factor $b(t)$, which is expanding at a linear rate with $t$ in both cases. However they disagree in the functional form of the scale factor $a(t)$, which is constant in the vacuum case, but forever expanding at a logarithmic rate with $t$ in the Klein-Gordon case. Because of this \emph{qualitative difference} we say that the Klein-Gordon future is \emph{matter dominated}.

We take this opportunity to emphasize a circumstance, which albeit having been implicitly taken into account in the literature, to our knowledge has not been pointed out explicitly so far. We discussed in~\textsc{Subsection}~\ref{SS: exact solutions} that one can identify the equilibrium points of the reduced dynamical system with a corresponding exact solution to the Einstein equations. Furthermore, in \textsc{Subsection}~\ref{SS: interpretation of eqp sols} we pointed out that if a solution to the reduced dynamical system converges to an equilibrium point, then this does \emph{not} necessarily imply that the associated metric converges to the exact solution associated with that point, but that the convergence rate to the point may have to be taken into account as well.

The question now is when the convergence rates are important and when not. We illustrate this at the example at hand. In our case, both the LRS Bianchi~III vacuum and Klein-Gordon solutions converge to the equilibrium point $D$, as can be seen from \textsc{Figure}~\ref{F: qualitative flow} and our analysis above. $D$ is associated with the Bianchi~III form of flat spacetime. However only the vacuum solutions are asymptotic to the latter. For the Klein-Gordon solutions we had to take into account the convergence rates and arrived at the qualitatively different result~\eqref{E: g asympt} in \textsc{Theorem}~\ref{T: g asympt}.

To illustrate why the convergence rates do not make a difference in the vacuum case, we go over it in the framework of our setup. From \textsc{Figure}~\ref{F: qualitative flow} and the proof of \textsc{Lemma}~\ref{L: attractors} we know that generic LRS Bianchi~III vacuum solutions lie on the stable manifold of $D$, and that the eigenvalue associated with that attraction is $-3/2$. By~\cite[\textsc Sec~2.9, Thm~1]{Perko2001} we then have $|\Sigma_+-1/2|\leq \kappa\exp(-3\tau/2)$ with $\kappa>0$. Transforming again~\eqref{E: a dot} to rescaled time, we have $a'=a(1-2\Sigma_+)$. Plugging into this the convergence estimate we get $a'=\mp 2\kappa a\exp(-3\tau/2)$, where the sign depends on whether we consider solutions with $\Sigma_+\gtrless1/2$. Integration yields
\begin{align}\notag
a(\tau) &= c_1 \exp\left( \mp\frac{4\kappa}{3}\mathrm{e}^{-3\tau/2} \right) \quad
\stackrel{\tau\approx\frac{2}{3}\ln t}{\Longrightarrow} \quad
a(t)\approx c_1 \exp\left(\mp\frac{4\kappa}{3 t}\right) \approx c_1 \quad\text{for large $t$}.
\end{align}
An analogous calculation can be done for $b$.

What we have shown is, that even if one takes into account the rate of approximation for the vacuum solutions, this does not alter the result from that associated with the equilibrium point solution. The reason is that the exponential convergence rate along the stable manifold resulted in the exponent in the equation for $a(t)$ to decay as $t^{-1}$ and thus to vanish for large $t$.

The general feature appears to be that the exponential convergence rate to hyperbolic sinks, (or on the stable manifolds of hyperbolic saddles or degenerate equilibrium points), results in the convergence rates to not alter the result from the exact solution associated with the fixed point. In cases of degenerate equilibrium points, the generally slower rate of approximation however can make a qualitative difference. In the vast majority of cases in the literature on spatially homogenous cosmology so far, the relevant equilibrium points have been hyperbolic. In these cases the natural association of the equilibrium point solution with the future asymptotic metric works. With degenerate equilibrium points one however has to be more careful as shows the case of~\cite{Rendall2002} as well as the case of the present work.

\section{Numerical support for the results}\label{S: numerical support}

In this section we test the agreement of our analytical results and assumptions against numerics. We give the numerical setup and equations in \textsc{Subsection}~\ref{SS: sim setup} and elaborate on the chosen time variable and simulation range. In subsections~\ref{SS: flow}--\ref{SS: future} we then present the simulation outcome in the form of flow plots (\textsc{Figure}~\ref{F: flow}) and function plots in normal (\textsc{Figure}~\ref{F: func}) and logarithmic scales (\textsc{Figure}~\ref{F: log}). The plots convincingly demonstrate agreement with our analytical results and assumptions. Finally, in \textsc{Subsection}~\ref{SS: past} we briefly investigate and comment on the possible past asymptotics at the hand of \textsc{Figure}~\ref{F: flow past}.

\subsection{Simulation setup}\label{SS: sim setup}

We simulate two systems:

Firstly we solve the guiding system~\eqref{E: guiding quasi-standard} of~\eqref{E: quasi-standard form} and simultaneously also integrate the differential form of~\eqref{E: t to tau} to capture the transformation between $t$ and $\tau$. Hence we can recover the solution of the original system~\eqref{E: quasi-standard form}, ie the quantities as functions of $t$. For a better overview we write out the full system explicitly:
\begin{align}
\partial_\tau H &= H[-(1+q)] &
\partial_\tau\Sigma_+ &= -(2-q)\Sigma_++1-\Sigma_+^2-\Omega \label{E: num S prime} \\
\partial_\tau t &= 1/H &
\partial_\tau\Omega &= 2\Omega \big(1+q-3\cos(t-\varphi)^2\big) \label{E: num O prime} \\
& &
\partial_\tau\varphi &= -3\sin(t-\varphi)\cos(t-\varphi) \label{E: num phi prime}
\end{align}
with $q$ given by~\eqref{E: r to O}. Note that as we work in the guiding system all quantities including $t$ are seen as functions of $\tau$.

Secondly, to better illustrate our results and the utility of the averaging method, we also solve an averaged system of~\eqref{E: num S prime}--\eqref{E: num phi prime} of the form $[\partial_\tau \overline H, \partial_\tau\overline{\mathbf x}]^\mathrm{T} = \overline{\mathbf F}^1(\overline{\mathbf x}) + \overline H\,\overline{\mathbf F}^{[2]}(\overline{\mathbf x})$ together with $\partial_\tau \overline t=1/\overline H$.\footnote{Note that here we understand $\mathbf F^1,\mathbf F^{[2]}$ to be averaged over a period of $2\pi$ in $t$, not in $\tau$.} The averaging reduces to substituting
\begin{align*}
q \mapsto \overline q:=2\overline\Sigma_+^2+\overline\Omega/2 \qquad
\cos(t-\varphi)^2 \mapsto 1/2 \qquad
\sin(t-\varphi)\cos(t-\varphi) \mapsto0
\end{align*}
in~\eqref{E: num S prime}--\eqref{E: num phi prime}. Written out explicitly the resulting system reads:
\begin{align}
\partial_\tau\overline H &= \overline H[-(1+\overline q)] &
\partial_\tau\overline\Sigma_+ &= -\Big(2\big(1-\overline\Sigma_+^2\big) - \tfrac{\overline\Omega}{2}\Big)\overline\Sigma_+ + 1-\overline\Sigma_+^2-\overline\Omega \label{E: num S prime av} \\
\partial_\tau\overline t &= 1/\overline H &
\partial_\tau\overline\Omega &= \overline\Omega\Big( 4\overline\Sigma_+^2 - \big(1-\overline\Omega\big)\Big) \label{E: num O prime av} \\
& &
\partial_\tau\overline\varphi &= 0 \label{E: num phi prime av}
\end{align}
Note the equivalence between~\eqref{E: guiding system} and the second equations of~\eqref{E: num S prime av},\eqref{E: num O prime av}.

As initial data we pick the six sets given in \textsc{Table}~\ref{Tbl: initial data} and refer to the corresponding solutions by the Roman numerals in the first column.
\begin{table}
\begin{tabular}{ l | l | l | l | l | l }
  solution	& $H(0)$	& $\Sigma_+(0)$	& $\Omega(0)$	& $\varphi(0)$	& $t(0)$	\\ \hline
  i		& 0.1		& 0.1				& 0.9			& 0			& 0			\\
  ii		& 0.1		& 0.4				& 0.1			& 0			& 0			\\
  iii		& 0.1		& 0.6				& 0.1			& 0			& 0			\\
  iv		& 0.02	& 0.48			& 0.02		& 0			& 0			\\
  v		& 0.1		& 0.48			& 0.02		& 0			& 0			\\
  vi		& 0.1		& 0.5				& 0.01		& 0			& 0			\\
\end{tabular} \vspace{.4 cm}
\caption{The six initial data sets for our simulations of the systems~\eqref{E: num S prime}--\eqref{E: num phi prime} and~\eqref{E: num S prime av}--\eqref{E: num phi prime av}. We essentially vary $(\Sigma_+(0), \Omega(0))$, ie the location of the initial data point projected into $\mathcal X$, except for iv where we also vary $H(0)$.}
\label{Tbl: initial data}
\end{table}
These data satisfy the Hamiltonian constraint~\eqref{E: hc} and thus, projected into the $\Sigma_+,\Omega$ plane, lie in $\mathcal X$. We integrate the averaged system~\eqref{E: num S prime av}--\eqref{E: num phi prime av} in the interval $\tau\in[-40,40]$ and plot the corresponding solutions in orange in the figures of this section. The solutions to the full system~\eqref{E: num S prime}--\eqref{E: num phi prime} are plotted in blue, and since we are primarily interested in the future asymptotics, we integrate these from $\tau=0$ to $\tau\approx10$ for all figures except \textsc{Figure}~\ref{F: flow past} where we also integrate them to $\tau=-40$. The cutoff at $\tau\approx10$ has the following reason: As will become clear from the discussion below, the solutions to the full system oscillate with exponentially increasing frequency; cf also \textsc{Lemma}~\ref{L: t to tau asympt}.\footnote{One could instead work with the variables as functions of $t$ in which case the frequency would stay constant. However, as can be seen from our plots, this benefit is traded against the necessity to integrate to very large times in order to reach the asymptotic regime.} Hence the numerical grid has to be refined to smaller and smaller step sizes. The cutoff occurs before roundoff errors become an issue, which is at $\tau\approx10$ for all solutions. Despite this limitation, our numerical results give a clear demonstration of the validity of our analytical results.

\subsection{Flow plots}\label{SS: flow}

\textsc{Figure}~\ref{F: flow 3D} shows a flow plot of the solutions in the $H,\Sigma_+,\Omega$ domain.
\begin{figure}
  \includegraphics[clip=true,trim=0 0 0 0,scale=.5]{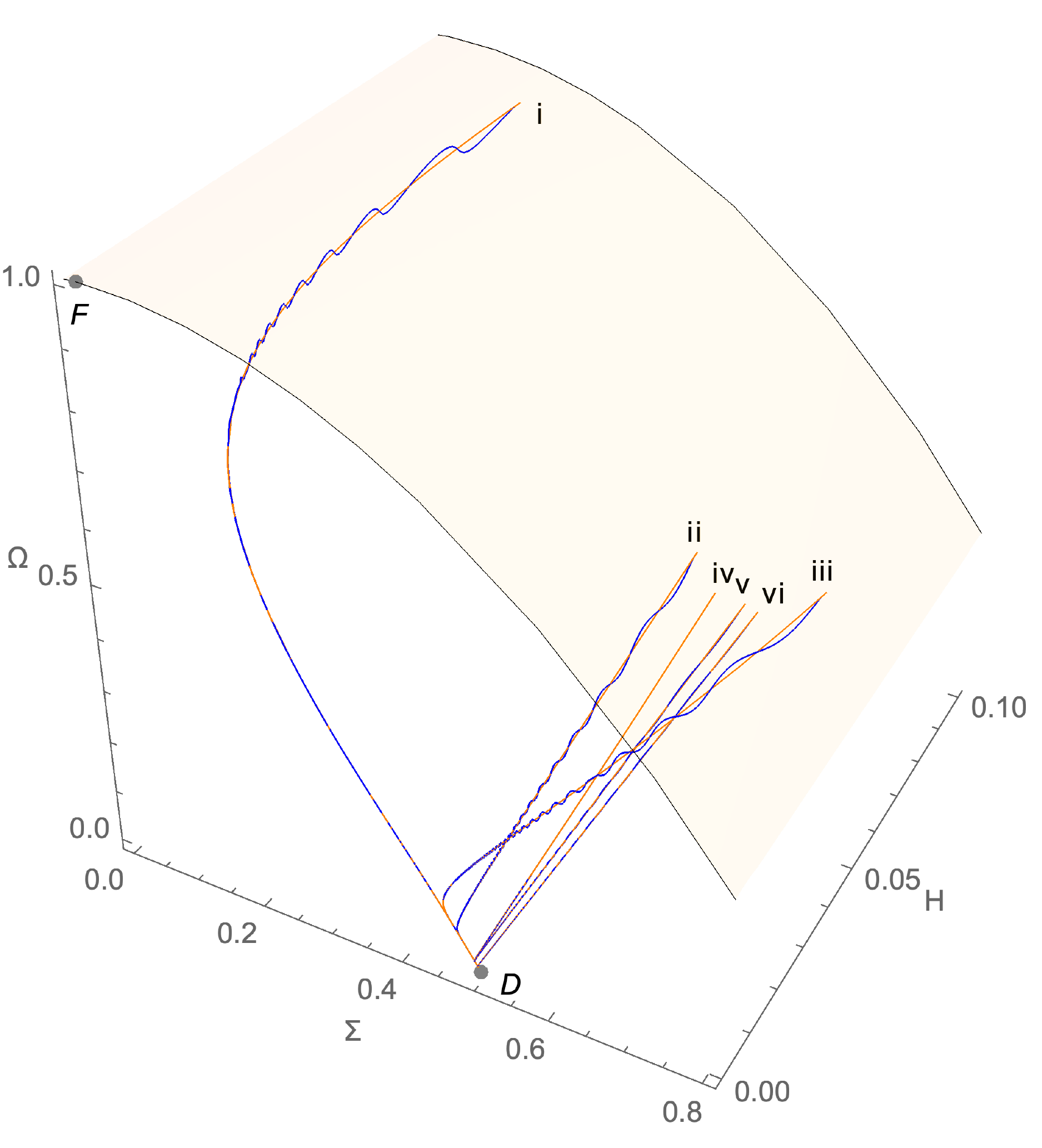}
  \caption{Flow plot of the solutions i--vi to both the full system~\eqref{E: num S prime}--\eqref{E: num phi prime} (blue) and the averaged system~\eqref{E: num S prime av}--\eqref{E: num phi prime av} (orange) in the $H,\Sigma_+,\Omega$ domain.}
  \label{F: flow 3D}
\end{figure}
The solutions to the full system oscillate around the solutions to the averaged system with a decaying envelope, and the plot demonstrates how all solutions converge to the point $D$ located at $(\Sigma_+,\Omega)=(1/2,0)$ in the $H=0$ plane. One can also see that, figuratively speaking, this convergence occurs in three steps as follows: Firstly the solutions approach the $H=0$ plane, then the centre manifold of $D$ and finally they approach $D$ itself along its centre manifold. On the one hand this demonstrates that the truncation from the full system to the first order approximation performed in \textsc{Subsection}~\ref{SS: 1st O approx}, as a consequence of which $H$ was approximated by a constant, was legitimate, and hence directly supports \textsc{Assumption}~\ref{AS: 1}. On the other hand, it directly supports \textsc{Lemma}~\ref{L: stability of D}.

\textsc{Figure}~\ref{F: flow} shows the projections of the trajectories into $H=0$.
\begin{figure}
  \begin{subfigure}{\textwidth}
    \centering
    \includegraphics[clip=true,trim=0 0 0 0,scale=.5]{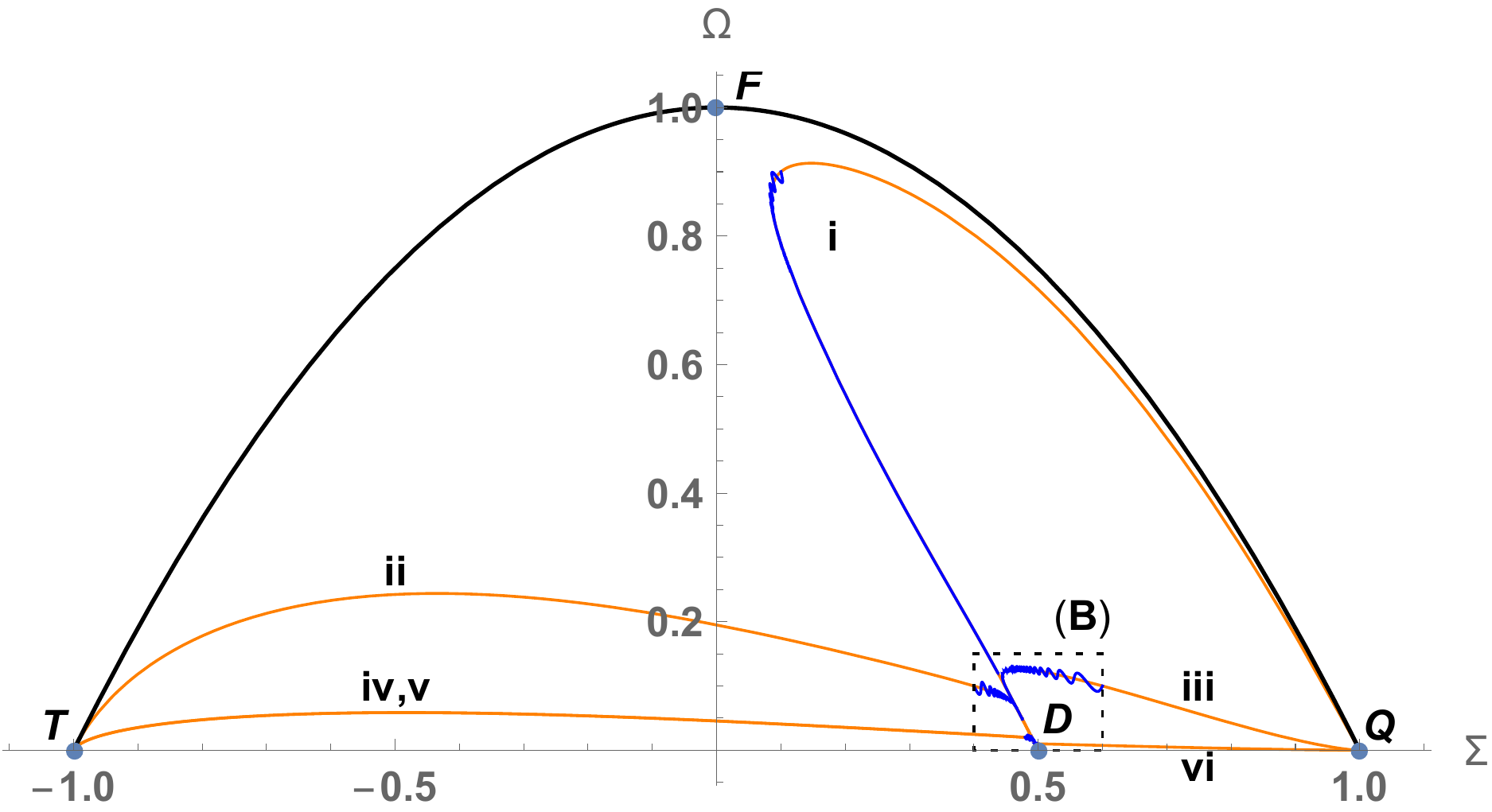}
    \caption{}
    \label{F: flow A}
  \end{subfigure}%
  \quad
  \\
   \begin{subfigure}{.48\textwidth}
    \centering
    \includegraphics[clip=true,trim=0 0 0 0,scale=.31]{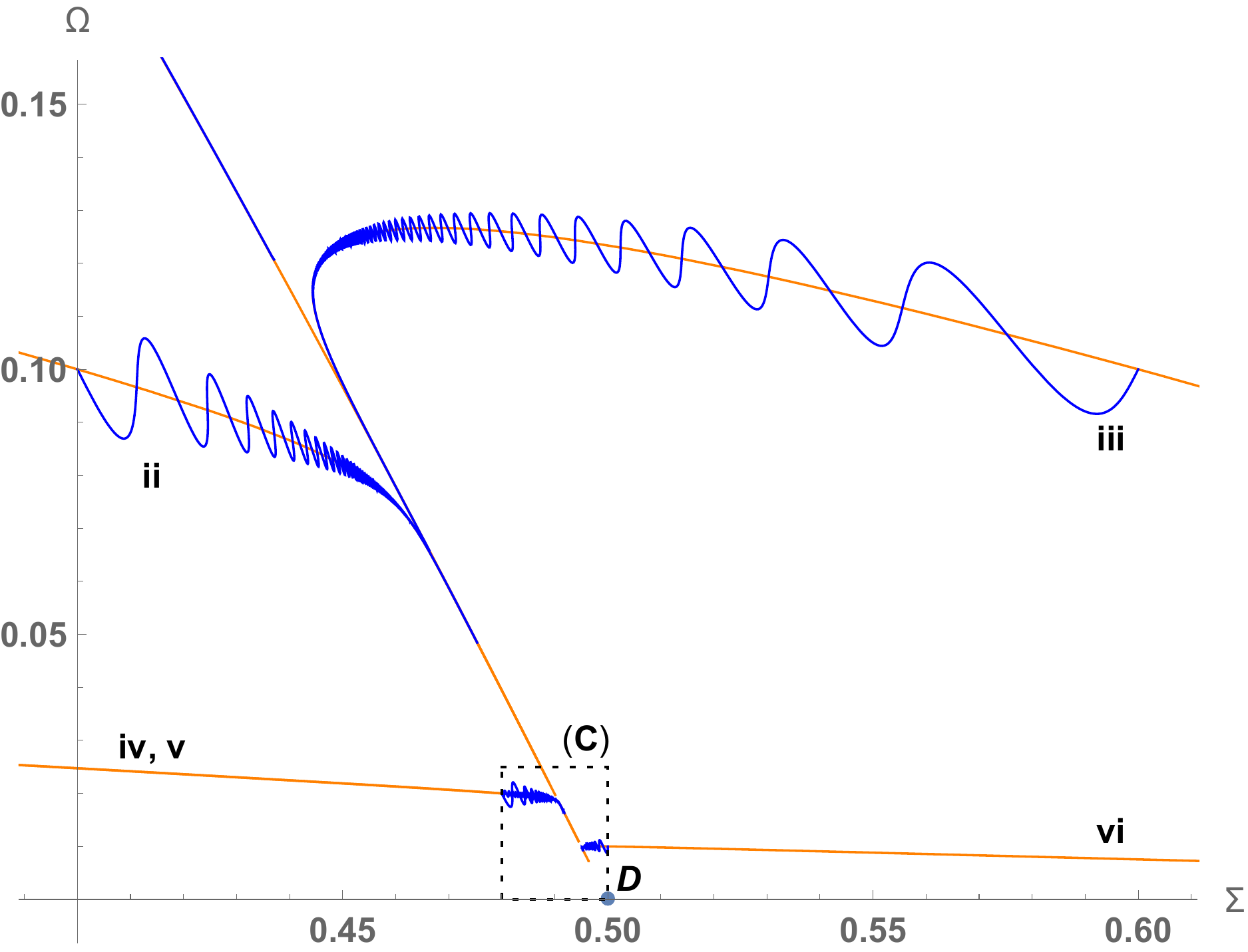}
    \caption{}
    \label{F: flow B}
  \end{subfigure}
  \quad
  \begin{subfigure}{.48\textwidth}
    \centering
    \includegraphics[clip=true,trim=0 0 0 0,scale=.31]{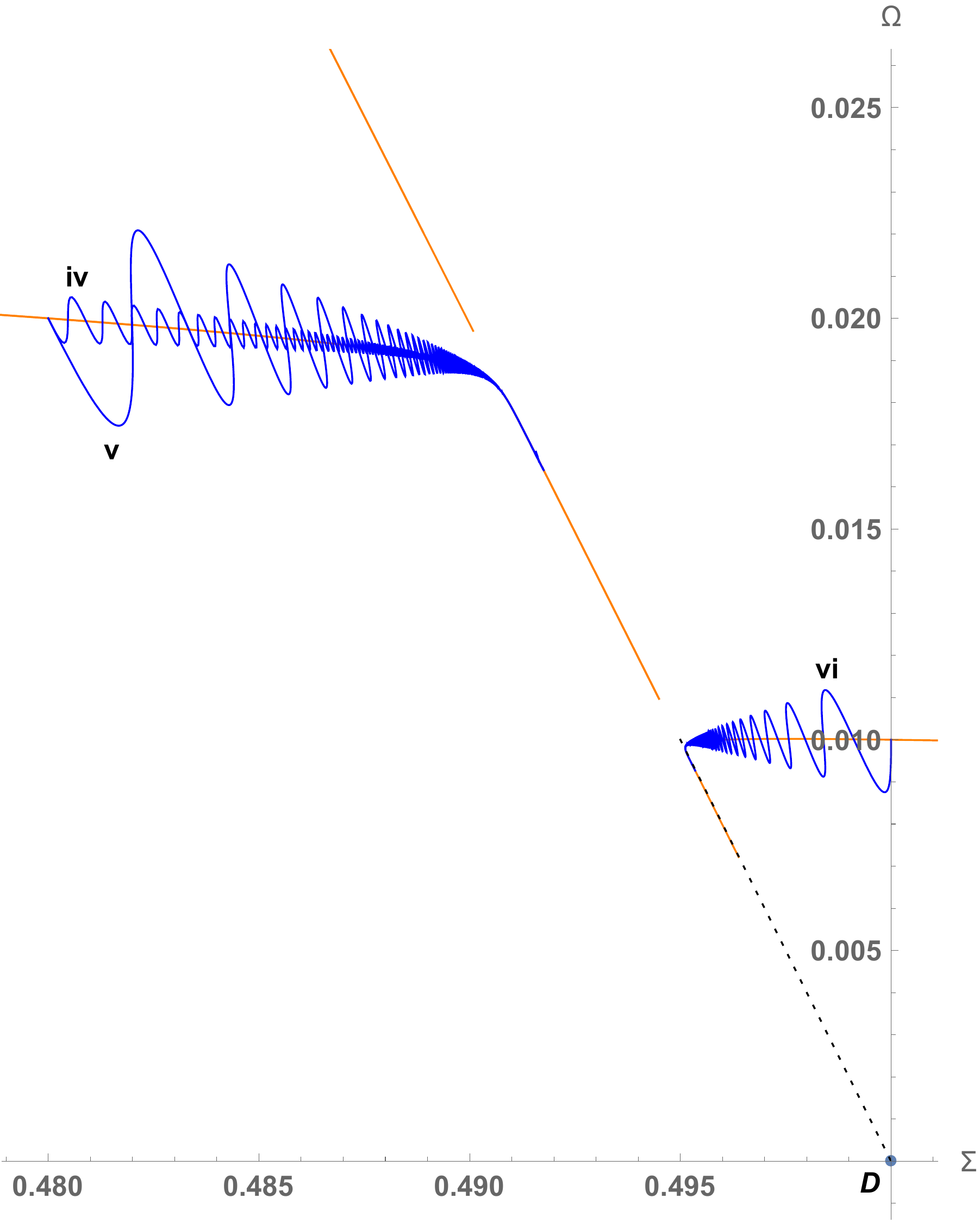}
    \caption{}
    \label{F: flow C}
  \end{subfigure}
\caption{Flow plots of the solutions i--vi to both the full system~\eqref{E: num S prime}--\eqref{E: num phi prime} (blue) and the averaged system~\eqref{E: num S prime av}--\eqref{E: num phi prime av} (orange) as projections into $H=0$. \textsc{Figure}~\textsc{\ref{F: flow A}} depicts the full state-space $\mathcal X$ while figures~\textsc{\ref{F: flow B}} and \textsc{\ref{F: flow C}} show zoomed in regions. The respective zoom windows are drawn into figures~\textsc{\ref{F: flow A}} and~\textsc{\ref{F: flow B}}.}
\label{F: flow}
\end{figure}
\textsc{Figure}~\textsc{\ref{F: flow A}} depicts the full state-space $\mathcal X$ while figures~\textsc{\ref{F: flow B}} and~\textsc{\ref{F: flow C}} show two zoomed in regions. The respective zoom boxes are drawn into figures~\textsc{\ref{F: flow A}} and~\textsc{\ref{F: flow B}}. The trajectories of the averaged solutions resemble the analytically obtained qualitative flow diagram in~\textsc{Figure}~\ref{F: qualitative flow}, which gives direct support of \textsc{Lemma}~\ref{L: attractors}. The trajectories of the solutions to the full system oscillate around the averaged trajectories with a decaying envelope. The initial data sets of solutions iv and v only differ by their values for $H(0)$; cf \textsc{Table}~\ref{Tbl: initial data}. The respective solutions are best viewed in \textsc{Figure}~\textsc{\ref{F: flow C}}, and the difference in initial amplitudes is consistent with $\mathcal O(0.1)$ and $\mathcal O(0.02)$, and thus with \textsc{Conjecture}~\ref{CJ: 1}. Again these plots nicely demonstrate how all solutions approach $D$ along its local centre manifold tangent to its eigenvector $[-1/2,1]^\mathrm{T}$ (dotted line in \textsc{Figure}~\textsc{\ref{F: flow B}}) and thus support \textsc{Lemma}~\ref{L: stability of D}. That this approach is slow in comparison to the speeds of the trajectories away from the local centre manifold can be seen as follows: While the endpoints of the trajectories to the full solutions are reached at $\tau\approx10$ the trajectories of the averaged solutions do not advance much further until they terminate at $\tau=40$, and this advance even becomes smaller and smaller the closer the initial data is to~$D$.

\subsection{Function plots}\label{SS: func}

\textsc{Figure}~\ref{F: func} restricts to solution~i and shows $H, \Sigma_+, \Omega$ and $t$ plotted against $\tau$ together with the respective averaged quantities, as well as separate error plots between the full and averaged solutions.
\begin{figure}
\begin{subfigure}{.48\textwidth}
    \centering
    \includegraphics[clip=true,trim=0 0 0 0,scale=.31]{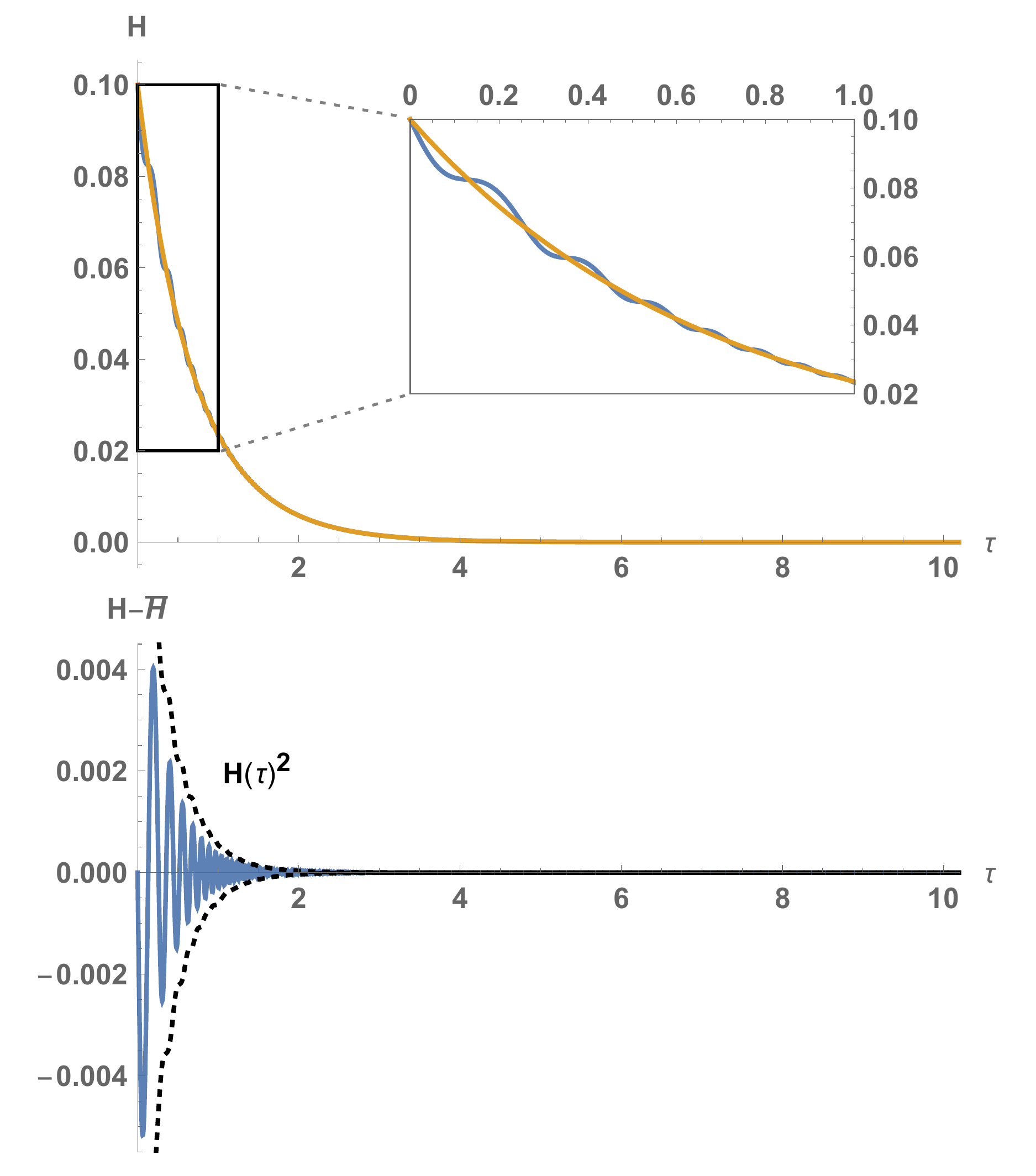}
    \caption{}
    \label{F: func H}
  \end{subfigure}
  \quad
  \begin{subfigure}{.48\textwidth}
    \centering
    \includegraphics[clip=true,trim=0 0 0 0,scale=.31]{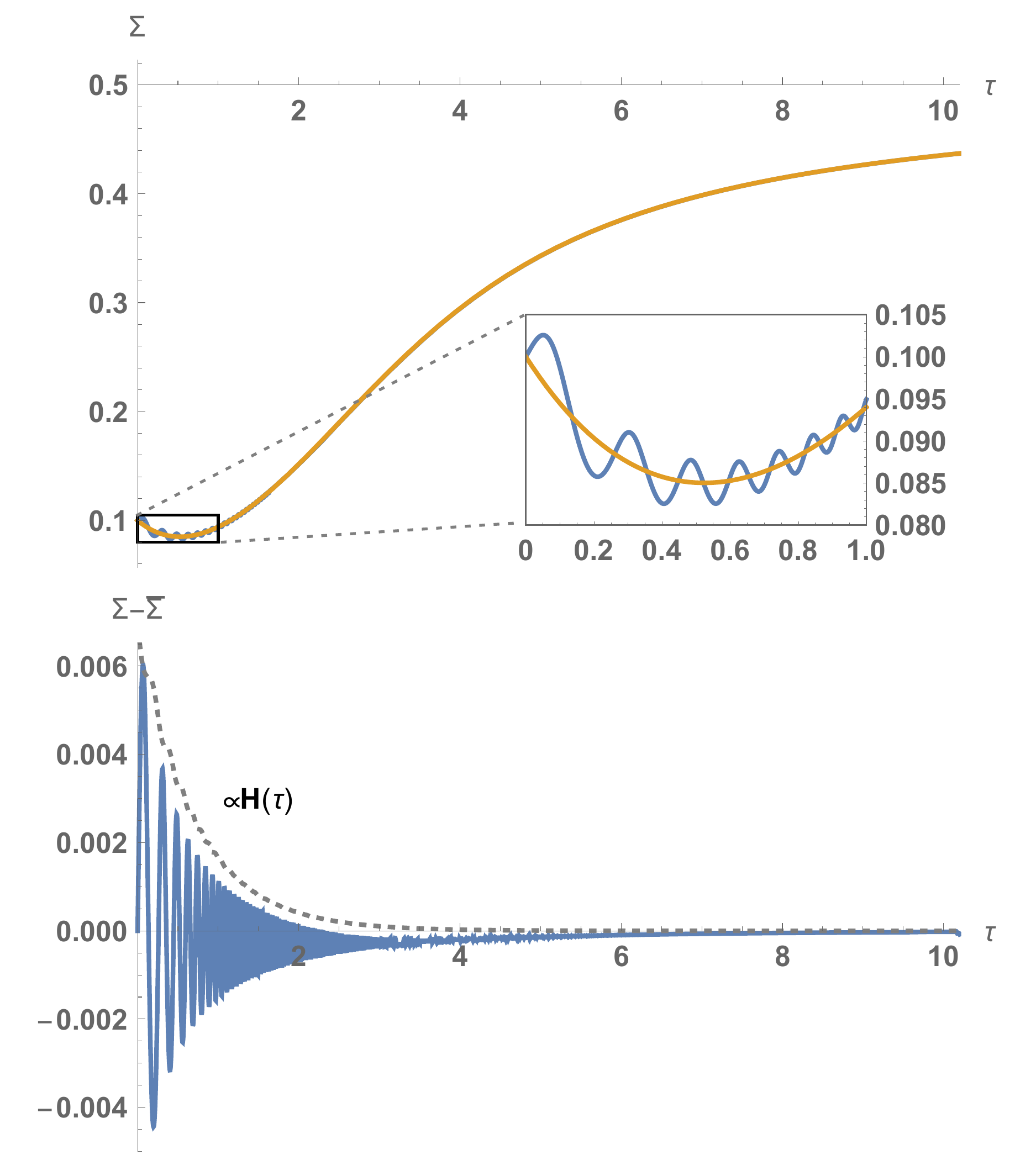}
    \caption{}
    \label{F: func S}
  \end{subfigure}
  \\
  \begin{subfigure}{.48\textwidth}
    \centering
    \includegraphics[clip=true,trim=0 0 0 0,scale=.31]{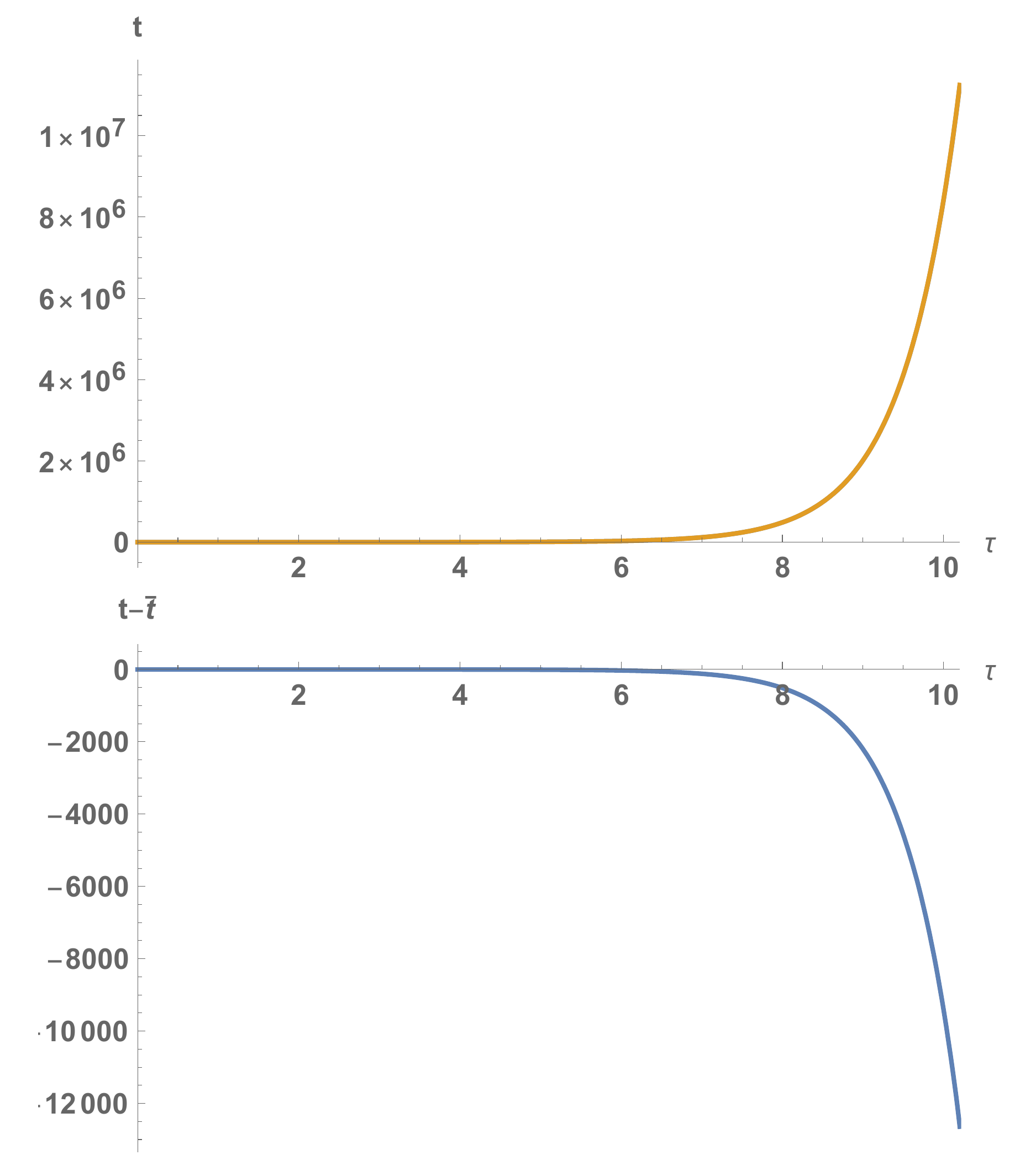}
    \caption{}
    \label{F: func t}
  \end{subfigure}
  \quad
  \begin{subfigure}{.48\textwidth}
    \centering
    \includegraphics[clip=true,trim=0 0 0 0,scale=.31]{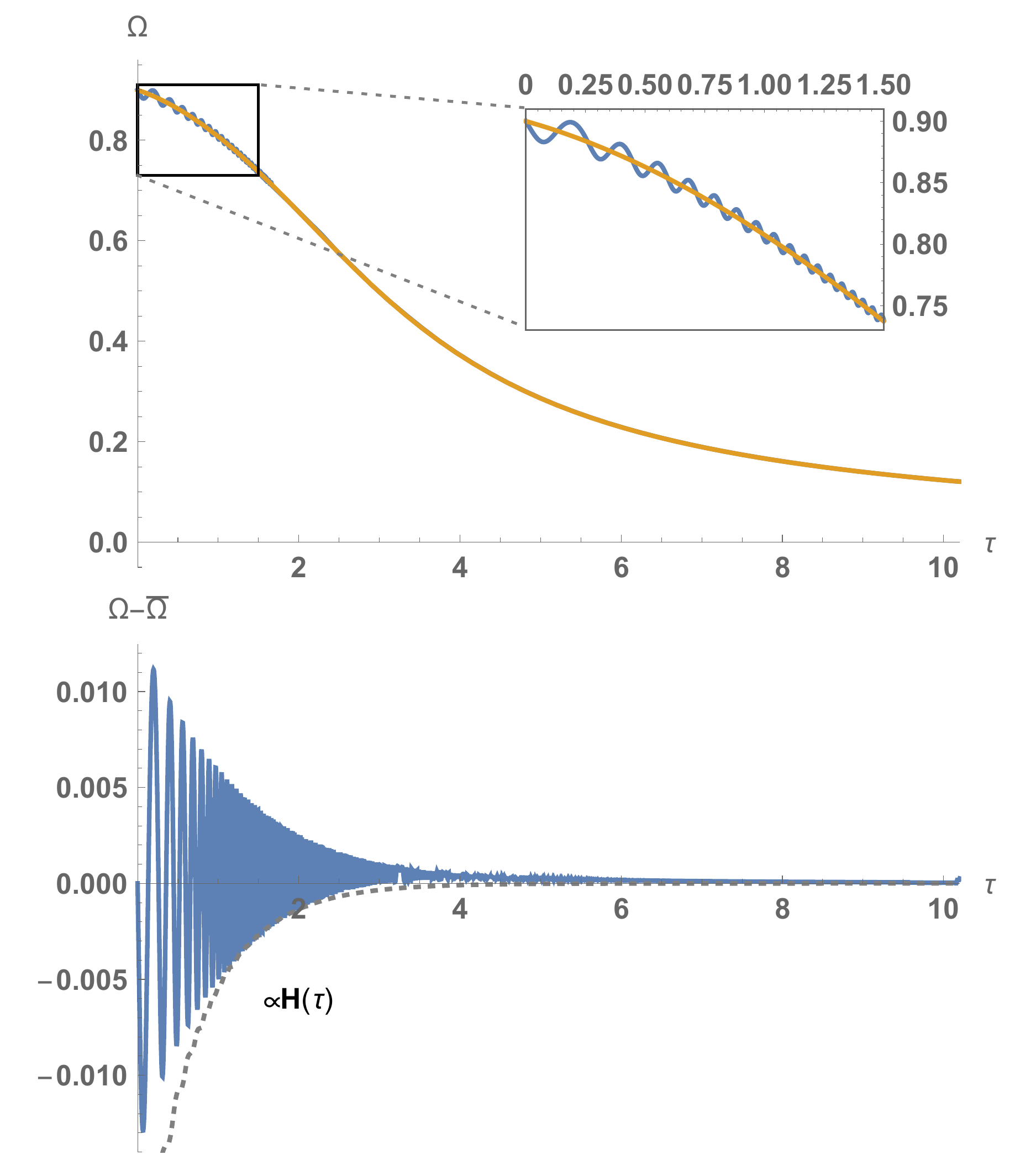}
    \caption{}
    \label{F: func O}
  \end{subfigure}
\caption{Plots of $H, \Sigma_+, \Omega$ and $t$ against $\tau$ (blue) for the solution i together with the respective averaged quantities (orange) and separate error plots between them (bottom plots in each subfigure).}
\label{F: func}
\end{figure}
The function plots of figures~\textsc{\ref{F: func H}}, \textsc{\ref{F: func S}} and \textsc{\ref{F: func O}} demonstrate nicely the oscillations of the full quantities on top of the averaged ones. The corresponding error plots demonstrate the fast convergence of the full to the averaged quantities. In particular, we see that the decay of the errors of $\Sigma_+, \Omega$ indeed follow an envelope proportional to $H(\tau)$ as conjectured in \textsc{Conjecture}~\ref{CJ: 1}.\footnote{There is a small offset of the mean errors of $\Sigma_+, \Omega$ from zero which converges away slower than $\mathcal O(H(\tau))$. In our numerical experiments we found that the magnitude of this error depends on the chosen initial phase $\varphi(0)$. We leave a deeper analysis of this phenomenon as an open problem.} The error of $H$ on the other hand appears to decay with $H(\tau)^2$, which we would await intuitively from the fact that the Raychaudhuri equation (left equation of~\eqref{E: num S prime}) is of second order in $H(\tau)$ and vanishes to first order.

Comparing \textsc{Figure}~\textsc{\ref{F: func H}} to figures~\textsc{\ref{F: func S}} and~\textsc{\ref{F: func O}} we have again confirmation of the faster decay of $H$ to~$0$ in comparison to the convergence rate of $(\Sigma_+, \Omega)$ to $D$. In other words, we have again confirmation that the truncation to the first order approximation is legitimate; cf~\textsc{Subsection}~\ref{SS: 1st O approx}.

The function plot in \textsc{Figure}~\textsc{\ref{F: func t}} shows the exponential growth of $t$ with $\tau$, supporting \textsc{Lemma}~\ref{L: t to tau asympt}. It is apparent from this plot that in a simulation in $t$ instead of $\tau$ one would have to evolve to very large times in order to enter the asymptotic regime, which we already pointed out in \textsc{Subsection}~\ref{SS: sim setup}. This is demonstrated even better in the following subsection.

\subsection{Future asymptotics}\label{SS: future}

In \textsc{Figure}~\ref{F: log} we show plots of the same quantities as in \textsc{Figure}~\ref{F: func}, however in logarithmic (figures~\textsc{\ref{F: log H}}, \textsc{\ref{F: log t}}) and double logarithmic scalings (figures~\textsc{\ref{F: log S}}, \textsc{\ref{F: log O}}), together with the analytically determined asymptotic behaviours (dashed).
\begin{figure}
\begin{subfigure}{.48\textwidth}
    \centering
    \includegraphics[clip=true,trim=0 0 0 0,scale=.31]{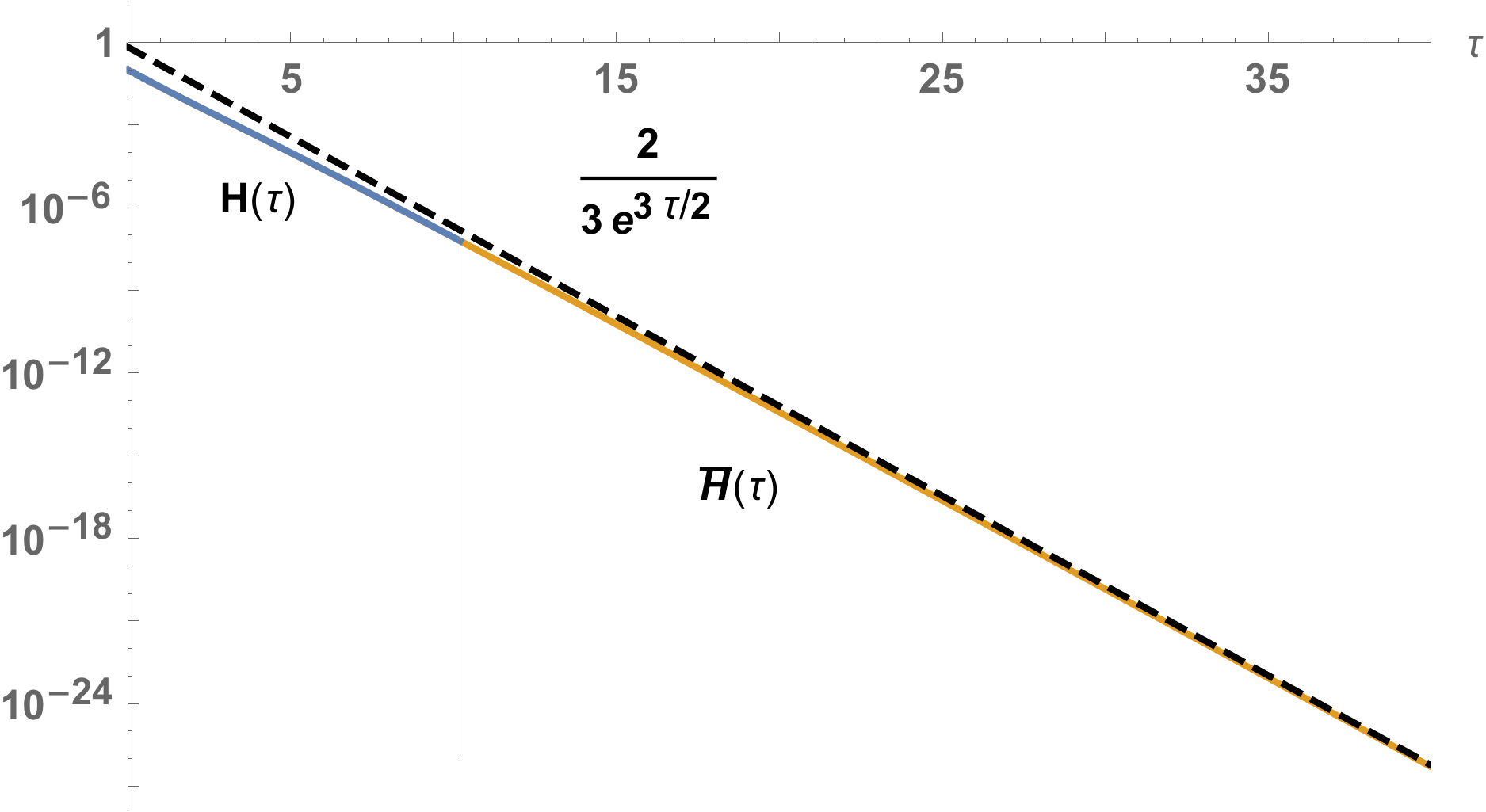}
    \caption{}
    \label{F: log H}
  \end{subfigure}
  \quad
  \begin{subfigure}{.48\textwidth}
    \centering
    \includegraphics[clip=true,trim=0 0 0 0,scale=.31]{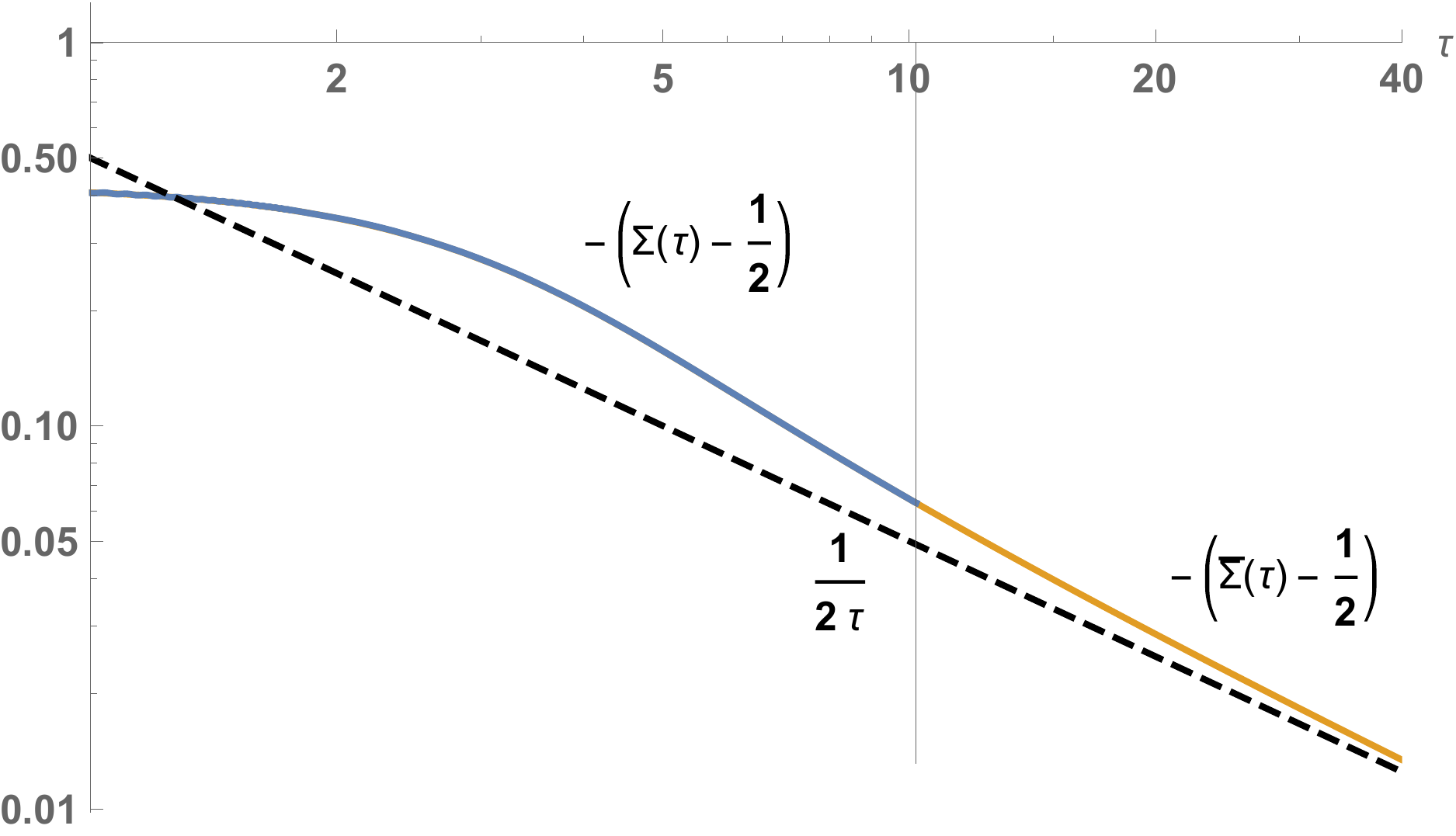}
    \caption{}
    \label{F: log S}
  \end{subfigure}
  \\
  \begin{subfigure}{.48\textwidth}
    \centering
    \includegraphics[clip=true,trim=0 0 0 0,scale=.31]{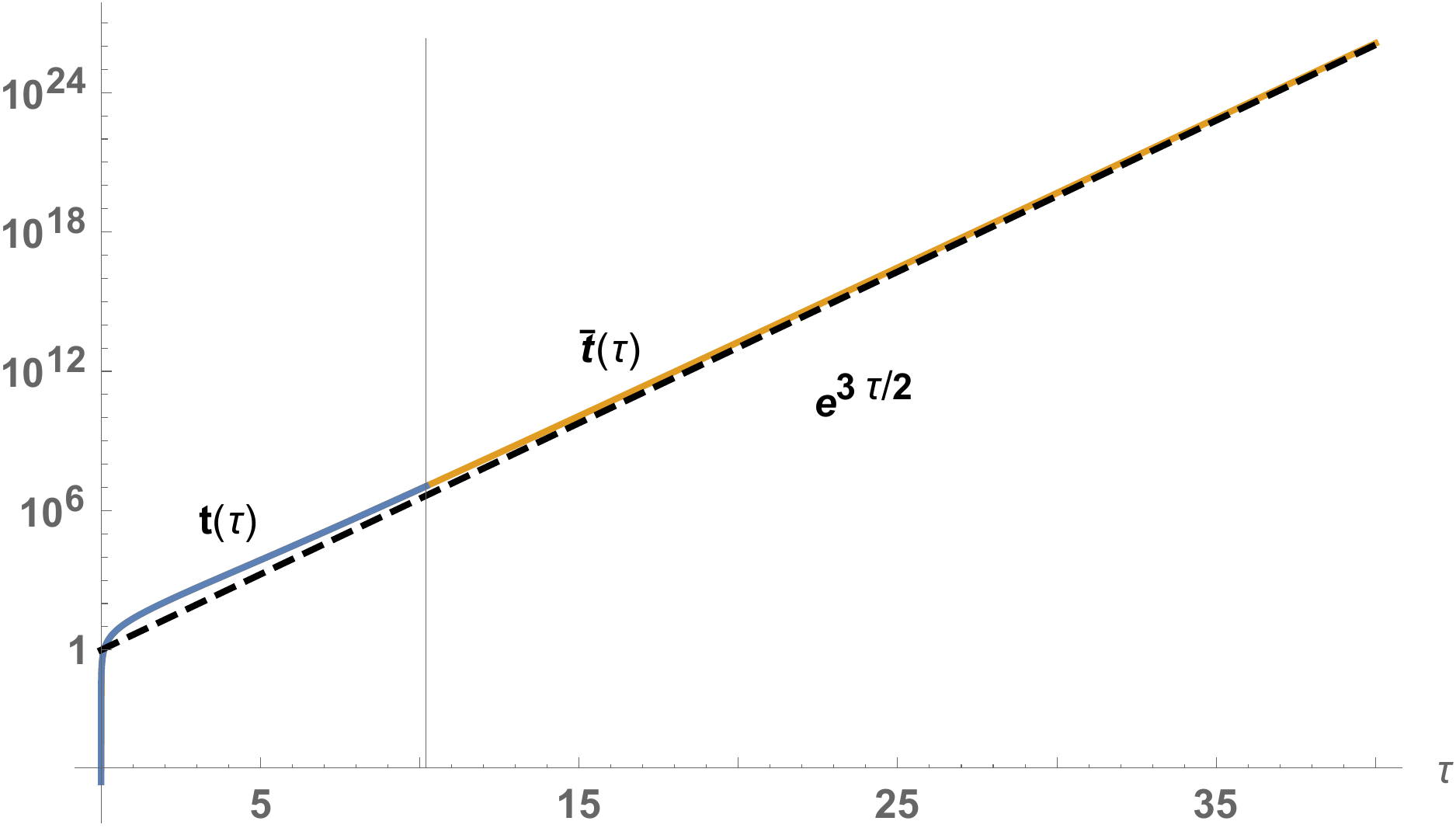}
    \caption{}
    \label{F: log t}
  \end{subfigure}
  \quad
  \begin{subfigure}{.48\textwidth}
    \centering
    \includegraphics[clip=true,trim=0 0 0 0,scale=.31]{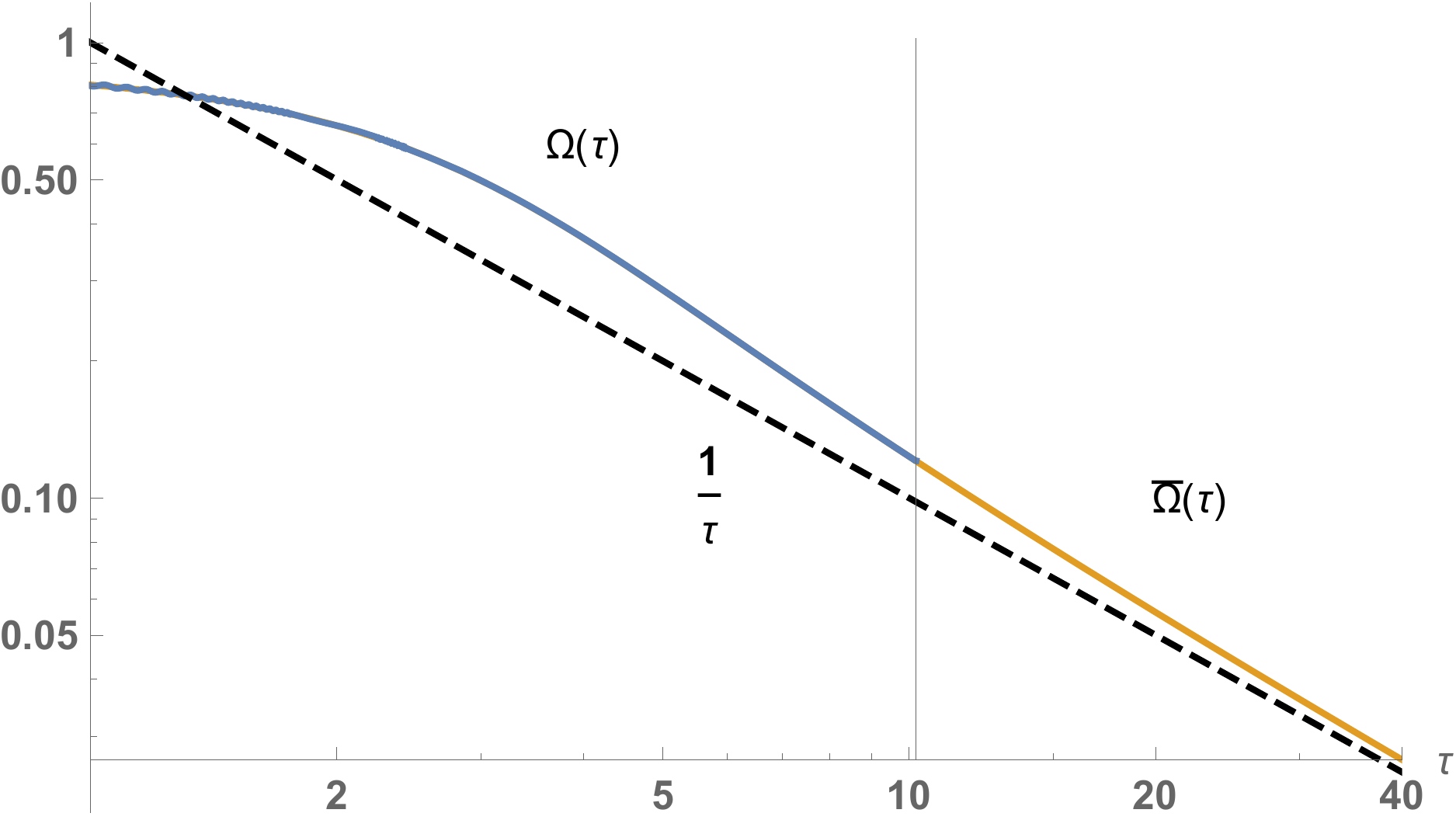}
    \caption{}
    \label{F: log O}
  \end{subfigure}
\caption{Plots of the same quantities as in~\textsc{Figure}~\ref{F: func}, however in logarithmic (\textsc{a, c}) or double logarithmic scaling (\textsc{b, d}), together with the analytically determined asymptotic behaviour (dashed). Both solutions have initial time $\tau=0$. The full solutions (blue) are plotted up to $\tau\approx10$, precisely to the vertical line. The averaged solutions (orange) are plotted up until $\tau=40$. The dashed lines show the analytically determined asymptotic behaviours.}
\label{F: log}
\end{figure}
Due to the logarithmic scaling, the asymptotic behaviours are determined by straight lines. This allows for a direct comparison of the numerics to the statements of \textsc{Theorem}~\ref{T: X asympt}, and for an indirect comparison to the statements of \textsc{Theorem}~\ref{T: g asympt}. Note also the higher plot range up to $\tau=40$. The full solutions are however still plotted up to $\tau\approx10$, where we terminated the integration. However this range is sufficient, especially in combination with the continuation of the averaged solution, to convincingly demonstrate the asymptotics.

From the plots it is apparent how all quantities indeed approach the respective analytically determined asymptotic behaviours for large $\tau$. It is also apparent that this happens faster for $H$ (and thus also for $t$) than for $\Sigma_+$ and $\Omega$, which we already stressed in the preceding subsections.

From (\textsc c) we stress again the very large values of $t$ required to enter the asymptotic regime. In fact, of all our solutions, solution~i plotted here turned out to be the one which reaches the asymptotic regime the fastest, which is why we chose it for figures~\ref{F: func} and~\ref{F: log}.

\subsection{Past asymptotics}\label{SS: past}

Though our focus here lies on the future asymptotics, it is also interesting to investigate what happens if we integrate our initial data backwards in time. In \textsc{Figure}~\ref{F: flow past} we thus show again the flow plot of \textsc{Figure}~\ref{F: flow}(\textsc a), however this time the full solutions are also integrated to $\tau=-40$ into the past.
\begin{figure}
  \includegraphics[clip=true,trim=0 0 0 0,scale=.5]{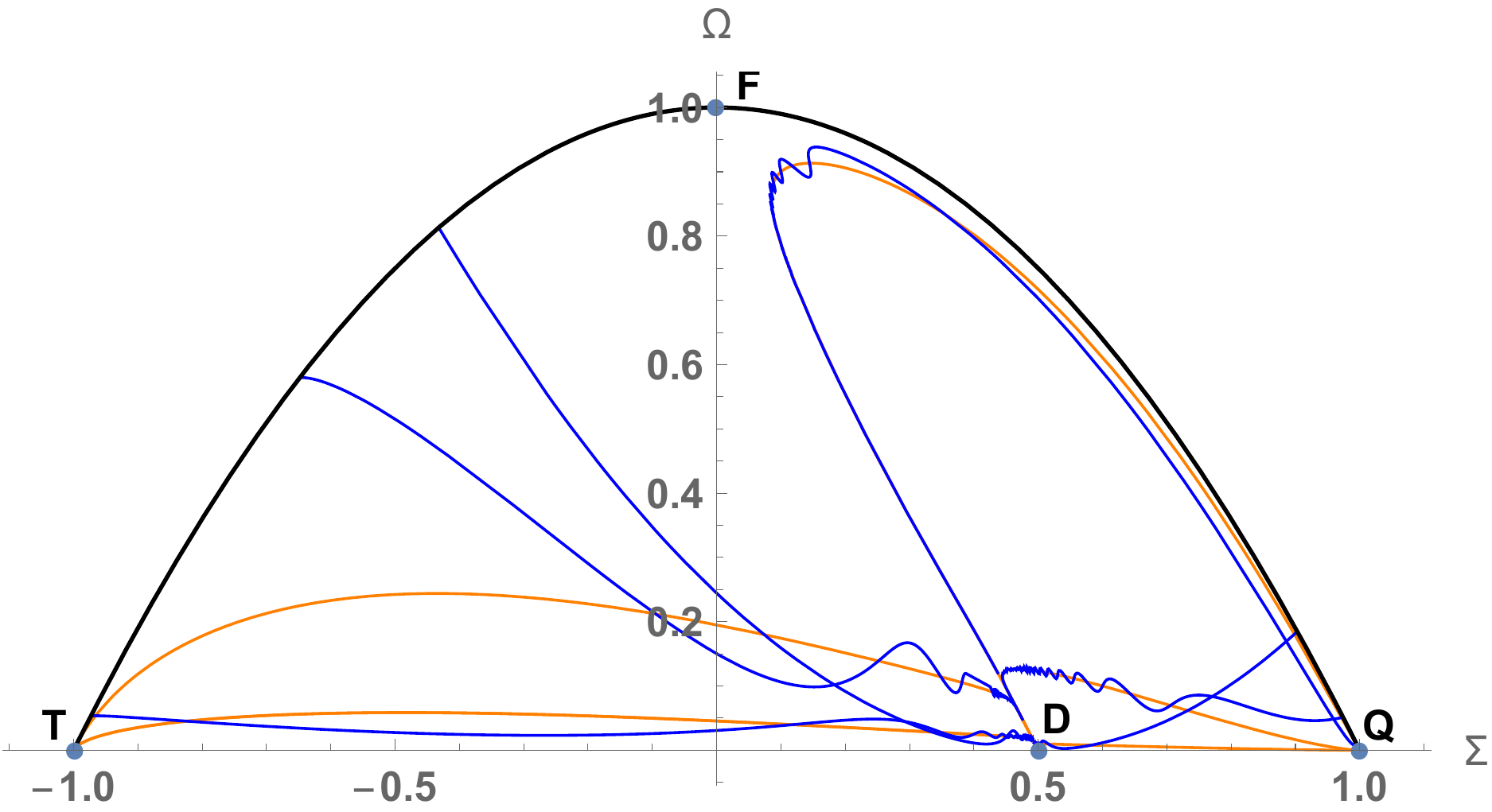}
  \caption{The flow plot of \textsc{Figure}~\ref{F: flow}(\textsc a), however this time the full solutions are also integrated to $\tau=-40$ into the past.}
\label{F: flow past}
\end{figure}

As we would expect since $H(\tau)$ is strictly monotonically decreasing, the oscillations around the averaged solutions become larger and larger into the past, and quickly reach the dimensions of the state-space. When this happens, each of our solutions approaches the Bianchi~I boundary. Hence one may conjecture that asymptotically into the past LRS Bianchi~III Einstein-Klein-Gordon solutions will behave like LRS Bianchi~I solutions, and potentially will converge to the LRS Kasner solutions associated with $T, Q$; cf \textsc{Subsection}~\ref{SS: interpretation of eqp sols}. However we have to be cautious with this intuition, since as stressed already beforehand, $\mathcal X$ does not represent the full reduced state-space. A deeper investigation is necessary in order to make a concrete statement about the past asymptotics. We leave this as an open question.

\section{Discussion}\label{S: discussion}

We conclude with a brief summary of our main results (\textsc{Subsection}~\ref{SS: summary}), some remarks on the novelty and utility of our averaging approach (\textsc{Subsection}~\ref{SS: utility of averaging}) and an outlook on open problems (\textsc{Subsection}~\ref{SS: outlook}).

\subsection{Summary of the main results}\label{SS: summary}

Our main results give the behaviour of LRS Bianchi~III Einstein-Klein-Gordon solutions for large times in terms of the shear variable and energy density (\textsc{Theorem}~\ref{T: X asympt}) and in terms of the metric components (\textsc{Theorem}~\ref{T: g asympt}). In addition \textsc{Theorem}~\ref{T: g asympt} also gives the late time behaviour of the Klein-Gordon field. The theorems have \textsc{Conjecture}~\ref{CJ: 1} as premise, for which we gave strong numerical and analytical support. The numerics also convincingly demonstrate agreement with the conclusions of the theorems.

Our results are global in the sense that \emph{all} solutions show this behaviour future asymptotically. The asymptotic metric we found is forever expanding in all spatial directions. The scale factor of the 2-hyperboloidal part of the spatial geometry expands $\propto t$ while the radial part only expands $\propto\ln t$. This result matches with that of~\cite{Rendall2002} where massive Vlasov matter was considered.

Due to the anisotropic expansion the shear variable is non-vanishing in the limit for time to infinity. The energy density on the other hand converges to vacuum. Despite that, we speak of the future asymptotics to be matter dominated, since it is qualitatively different from the future asymptotics in the vacuum case, where the first scale factor goes to a constant; cf~\textsc{Subsection}~\ref{SS: vacuum case}.

Consistent with the energy density going to zero, the Klein-Gordon field asymptotically oscillates with decaying envelope and uniform frequency, and goes to zero. This was to expect from the point of view that the Klein-Gordon equation in our case has the form of a non-linearly damped harmonic oscillator.

\subsection{Utility of the averaging approach}\label{SS: utility of averaging}

Apart from our main results, another important aspect of our work is the approach we took. The application of averaging methods allowed us to control the oscillations stemming from the Klein-Gordon equation and at the same time resulted in a dynamical system with decoupled Raychaudhuri equation such that then standard methods of dynamical systems in mathematical cosmology could be applied.

Due to the decay of the perturbation parameter to zero, ie of the Hubble scalar, and since the averaged system is attracted by an equilibrium point, the oscillations in the full solution decay and vanish in the limit for time to infinity. Consequently, the full and the averaged solutions converge in that limit. Our results for the limit thus concern the solutions to the full system and are not a mere averaging approximation.

To our best knowledge, closely related techniques have yet only been developed and applied to this branch of mathematical cosmology in~\cite{AlhoUggla2015, AlhoEtAl2015, AlhoEtAl2019}. With the present work we hope to draw more attention to the utility of averaging approaches in mathematical cosmology.

\subsection{Outlook}\label{SS: outlook}

We expect these techniques to also be applicable to other classes of spatially homogenous Einstein-Klein-Gordon cosmologies. While it would of course complicate the dynamical systems analysis of the averaged system, a larger number of degrees of freedom would not introduce any additional conceptual difficulties to the averaging approach as such.

With respect to the matter model, a natural next step would be to consider self-interaction potentials of order higher than quadratic, such as $\phi^4$. It would be interesting to see how respective results compare to those of the present work.

Furthermore, averaging techniques should also be a useful tool in other cases where oscillations or other perturbations play a role. For instance, we found striking similarities between the system discussed here, and the evolution equations encountered in \cite{HorwoodEtAl2003, NilssonEtAl2000, WainwrightEtAl1999}. These papers discuss the future asymptotics of different classes of spatially homogenous perfect fluid cosmologies, and encountered oscillatory behaviour and a phenomenon called self-similarity breaking; cf~\cite{Wainwright2000}. Averaging methods might be suited to shed new light on this phenomenon.

Though our results are robust given the strong agreement with numerics and the strong analytical support, it is desirable to rigorously prove \textsc{Conjecture}~\ref{CJ: 1}. \textsc{Section}~\ref{S: analytical support} suggests one route to do so, which is via providing rigorous proofs to \textsc{Assumption}~\ref{AS: 1} and the limiting process outlined at the end of that section.


\begin{appendix}

\numberwithin{equation}{section}
\renewcommand{\theequation}{\thesection\arabic{equation}}

\section{Two definitions}\label{A: two definitions}

We take the following two definitions from~\cite[p~31]{SandersEtAl2007} and~\cite[\textsc{Def}~4.2.4]{SandersEtAl2007} respectively.

\begin{AppDef}\label{D: D}
$D\subset\R^n$ is a connected, bounded open set (with compact closure) containing the initial value $\mathbf a$, and constants $L>0, \epsilon_0>0$, such that the solutions $\mathbf x(t,\epsilon)$ and $\mathbf z(t,\epsilon)$ with $0\leq\epsilon\leq\epsilon_0$ remain in $D$ for $0\leq t\leq L/\epsilon$.
\end{AppDef}
See also the comments on~\cite[p~31]{SandersEtAl2007} how such a tripple $(D, \epsilon_0, L)$ can be chosen.
\begin{AppDef}\label{D: KBM}
Consider the vector field $\mathbf f(\mathbf x, t)$ with $\mathbf f:\R^n\times\R\to\R^n$. Let $\mathbf f$ be Lipschitz continuous in $\mathbf x$ on $D\subset\R^n,t\geq0$. Let further $\mathbf f$ be continuous in $t$ and $\mathbf x$ on $\R^+\times D$. If the average
\begin{align*}
\overline{\mathbf f}(\mathbf x) &= \lim_{T\to\infty} \frac{1}{T}\int^T_0\mathbf f(\mathbf x, s)\mathrm ds
\end{align*}
exists and the limit is uniform in $\mathbf x$ on compact subsets of $D$, then $\mathbf f$ is called a \textbf{KBM-vector field} (Krylov, Bogoliubov and Mitropolsky).

(If the vector field $\mathbf f(\mathbf x, t)$ contains parameters we assume that the parameters and the initial conditions are independent of $\epsilon$, and that the limit is uniform in the parameters.)
\end{AppDef}

\section{Centre manifold analysis}\label{A: CM analysis}

Here we closely follow the discussion in~\cite[\textsc{Chap~1}]{Carr1981}. Given an autonomous dynamical system $\partial_\tau x=f(x), x\in\R^n$ a set $\mathcal S\subset\R^n$ is a \emph{local invariant manifold} for the system if for initial data $x(0)=x_0\in\mathcal S$ the corresponding solution $x(\tau)$ is in $\mathcal S$ for $|\tau|<T$ where $T>0$. If we can always choose $T=\infty$, then we say that $\mathcal S$ is an \emph{invariant manifold}.

Consider now the system
\begin{align}
\partial_\tau x &= Ax + f(x,y), x\in\R^n, f\in\mathcal C^2 \label{E: CM x dot} \\
\partial_\tau y &= By + g(x,y), y\in\R^m,g\in\mathcal C^2 \label{E: CM y dot}
\end{align}
with constant matrices $A,B$ such that all eigenvalues of $A$ have zero real parts while all eigenvalues of $B$ have negative real parts. Further $f(0,0)=0, f'(0,0)=0, g(0,0)=0, g'(0,0)=0$, where $f',g'$ denote the Jaccobians. Note that the system exhibits an equilibrium point at the origin.

If $y=h(x)$ is a (local) invariant manifold of~\eqref{E: CM x dot}--\eqref{E: CM y dot} and $h$ is smooth, then it is called a \emph{(local) centre manifold} if $h(0)=0,h'(0)=0$. Where it is clear from the context, we also speak of a centre manifold and really mean a local centre manifold.
\begin{AppLem}\label{T: CM existence}
There exists a centre manifold for~\eqref{E: CM x dot}--\eqref{E: CM y dot}, $y=h(x), |x|<\delta$, where $h$ is $\mathcal C^2$.
\end{AppLem}
The flow on the centre manifold is governed by the $n$ dimensional system
\begin{align}
\partial_\tau u &= Au + f(u,h(u)). \label{E: flow on CM}
\end{align}
By the following theorem~\eqref{E: flow on CM} contains all the necessary information to determine the asymptotic behaviour of small solutions of~\eqref{E: CM x dot}--\eqref{E: CM y dot}. For the different notions of stability used in the theorem we refer to~\cite[\textsc{Sec~2.9}]{Perko2001}
\begin{AppLem}\label{T: CM stability and flow}
This lemma consists of two parts:
\begin{enumerate}[(a)]
\item Suppose that the zero solution of~\eqref{E: flow on CM} is stable (asymptotically stable) (unstable). Then the zero solution of~\eqref{E: CM x dot}--\eqref{E: CM y dot} is stable (asymptotically stable) (unstable). \label{TI: a}
\item Suppose that the zero solution of~\eqref{E: flow on CM} is stable. Let $(x(\tau),y(\tau))$ be a solution of~\eqref{E: CM x dot}--\eqref{E: CM y dot} with $(x(0),y(0))$ sufficiently small. Then there exists a solution $u(\tau)$ of~\eqref{E: flow on CM} such that as $t\to\infty$, \label{TI: b}
\begin{align}
x(\tau)&=u(\tau)+\mathcal O(e^{-\gamma \tau}) \\
y(\tau)&=h(u(\tau))+\mathcal O(e^{-\gamma \tau})
\end{align}
where $\gamma>0$ is a constant.
\end{enumerate}
\end{AppLem}
In other words, solutions of~\eqref{E: CM x dot}--\eqref{E: CM y dot} with initial value sufficiently close to the origin, for large times approach solutions of~\eqref{E: flow on CM}, ie on the centre manifold, at an asymptotic rate.

To solve~\eqref{E: flow on CM} we need to know the centre manifold $y=h(x)$ or at least an approximation of it for small $|x|$. Substituting $y=h(x(\tau))$ into~\eqref{E: CM y dot} yields
\begin{align}
h'(x)\big(Ax+f(x,h(x))\big) &= Bh(x) + g(x,h(x)). \label{E: DE for h}
\end{align}
For the centre manifold we have $h(0)=0, h'(0)=0$. Consequently, to obtain the centre manifold one has to solve~\eqref{E: DE for h} to these conditions. In general this is impossible analytically. However, with the following theorem we can always approximate the centre manifold up to an arbitrarily high accuracy. Before we formulate the theorem we define the following map $M$ on functions $\phi:\R^n\to\R^m$ which are $\mathcal C^1$ in a neighbourhood of the origin:
\begin{align}\label{E: M map}
(M\phi)(x):=\phi'(x)\big(Ax+f(x,\phi(x))\big) - B\phi(x)-g(x,\phi(x)).
\end{align}
Note that by~\eqref{E: DE for h}, $(Mh)(x)=0$.
\begin{AppLem}\label{T: CM approx}
Let $\phi$ be a $\mathcal C^1$ mapping of a neighbourhood of the origin in $\R^n$ into $\R^m$ with $\phi(0)=0$ and $\phi'(0)=0$. Suppose that as $x\to0$, $(M\phi)(x)=\mathcal O(|x|^q)$ where $q>1$. Then as $x\to0$, $|h(x)-\phi(x)|=\mathcal O(|x|^q)$.
\end{AppLem}

\end{appendix}

\bibliography{BibFile}{}
\bibliographystyle{siam}

\end{document}